\documentclass[11pt]{article}

\usepackage[margin=1in]{geometry}
\usepackage{amsmath, amssymb, amsthm}
\usepackage{graphicx}
\usepackage{tabularx}
\usepackage{tikz}
\usetikzlibrary{calc,decorations.pathreplacing}
\usepackage{booktabs}
\usepackage{multirow}
\usepackage{hyperref}
\usepackage{natbib}
\usepackage{algorithm}
\usepackage{enumitem}
\usepackage{subcaption}
\usepackage{times}
\usepackage{appendix}
\usepackage{placeins}

\usetikzlibrary{arrows.meta, shadows.blur, calc}
\usepackage{xcolor}
\definecolor{coldblue}{RGB}{214,234,248}
\definecolor{coldframe}{RGB}{52,120,180}
\definecolor{hotred}{RGB}{248,215,210}
\definecolor{hotframe}{RGB}{185,70,60}
\definecolor{neutralbg}{RGB}{240,245,240}
\definecolor{neutframe}{RGB}{90,135,100}
\definecolor{arrowcol}{RGB}{70,90,110}
\definecolor{ceilcol}{RGB}{25,95,165}
\definecolor{floorcol}{RGB}{180,45,35}
\definecolor{outdoororange}{RGB}{220,100,20}
\definecolor{indoorcyan}{RGB}{20,150,180}

\usepackage{nomencl}
\usepackage{multicol}
\makenomenclature

\renewcommand{\nompreamble}{\begin{multicols}{2}}
	\renewcommand{\nompostamble}{\end{multicols}}

\DeclareMathOperator*{\argmin}{arg\,min}
\DeclareMathOperator*{\argmax}{arg\,max}
\newcommand{\tsum}{\sum\limits}

\usepackage{pgffor}
\usepackage{setspace}

\title{Robust inference for cyclic-stress accelerated life tests under interval monitoring with lognormal lifetimes}
\author{
	María Jaenada\\
	\small Department of Statistics and O.R., UNED, Madrid, Spain
    \and Leandro Pardo \\
	\small Department of Statistics and O.R., Complutense University of Madrid, Spain
	\and
    Kiran Prajapat\thanks{Corresponding author. Email: kiranprajapat92@gmail.com, kiran.prajapat@newcastle.ac.uk (Kiran Prajapat)}\\
	\small School of Mathematics, Statistics and Physics, Newcastle University, Newcastle upon Tyne, UK
}\date{}

\begin{document}
	
	\maketitle
	\begin{abstract}

    Highly reliable products are often tested under accelerated conditions to provoke failures within a feasible timeframe. For products whose service life involves repeated alternation between two stress levels, such as automotive air-conditioners, batteries, and aerospace components, cyclic-stress accelerated life testing (CyALT) provides a more realistic loading profile than conventional accelerated tests. In practice, failures are often recorded only at scheduled inspection times, leading to interval-censored counts rather than exact lifetimes. Moreover, traditional maximum likelihood estimation is sensitive to data contamination, which is a genuine concern in small-sample industrial experiments. This paper develops robust inferential procedures for CyALT models with lognormal lifetimes under interval monitoring. Robust estimators are obtained by minimizing a weighted density power divergence (WDPD), leading to the weighted minimum density power divergence estimator (WMDPDE). We establish the asymptotic distribution of the WMDPDE, derive influence function expressions to characterize the robustness, and present asymptotic and bootstrap confidence intervals for important lifetime characteristics. A simulation study confirms that the WMDPDE provides substantial protection against outliers while retaining high efficiency under clean data. The methodology is illustrated through the analysis of an air-conditioner reliability dataset, demonstrating the practical advantages of robust inference in the CyALT framework.

	\end{abstract}
	\providecommand{\keywords}[1]{\textbf{Keywords:} #1}
	\\[2pt]
	
	\newtheorem{theorem}{Theorem}[section]
	\newtheorem{corollary}{Corollary}[theorem]
	\newtheorem{proposition}{Proposition}[section]
	\newtheorem{lemma}{Lemma}[section]
	\newtheorem{remark}{Remark}[section]
	\newtheorem{result}{Result}[section]
	
	\section{Introduction}
	
	Testing of highly reliable products under normal operating conditions is often impractical in the sense that units may run for years without a single failure, making it impossible to gather useful information on lifetime behavior within an affordable timeframe. Accelerated life testing (ALT) addresses this by exposing testing units to higher-than-normal stress levels to provoke earlier failures, with reliability at use conditions then estimated by extrapolation. There are three stress loading mechanisms that dominate the ALT literature: constant-stress, step-stress, and progressive-stress. In constant-stress ALTs, each unit is assigned to a single fixed stress level for the duration of the test \citep{meeker1998, bai1991}. Whereas step-stress ALTs increase the stress at pre-specified time points, producing more failures more quickly under comparable experimental budgets \citep{nelson1980, miller1983, bai1989}. Moreover, progressive-stress ALTs apply a continuously increasing stress during the experiment. Thorough treatments of the statistical models, planning, and analysis for all three ALT designs can be found in \cite{nelson1990} and \cite{meeker1998}.

    A common assumption underlying constant-stress and step-stress ALTs is that the stress profile remains constant within each assigned period. This is adequate for many products, but not for those whose actual service/working conditions involve repeated alternation between two stress levels. Automotive air-conditioners, for instance, cycle continuously between high-pressure and low-pressure states as the compressor and expansion valve switch on and off to expand and drop the pressure levels. Similar cyclic operating profiles arise in batteries (charge-discharge cycles), aerospace mechatronic components (day-night thermal cycling), and firefighter protective gear (repeated washing and drying). For such products, a cyclic-stress ALT (CyALT), in which each unit alternates periodically between a ceiling and a floor stress level, more faithfully reflects in-service loading while still providing the acceleration necessary for making the timely lifetime characteristics' inference \citep{wang2024cyclic}. \cite{kim2021optimal} initiated this line of work by developing optimal CyALT plans under lognormal lifetimes, Type-I censoring, and the cumulative exposure model, determining the common floor level and unit allocation proportions to minimize the asymptotic variance of the MLE of a lifetime quantile at the use condition. \cite{kim2023optimal} subsequently extended this framework to the complete data case, developing both optimal and compromise CyALT plans and providing practical sample size guidelines. More recently, \cite{zhang2024reliability} studied reliability analysis under cyclic ALTs for the log-location-scale family of distributions under censoring, and \cite{zhang2025reliability} considered reliability estimation under cyclic ALTs based on a scale family of distributions.

    In spite of these advances, there are two important gaps that remain in the CyALT literature. The first concerns the form of the observed data. In many industrial settings, continuous monitoring of each testing unit is not feasible; units can only be inspected at scheduled time points, so the data consist of failure counts within intervals rather than exact lifetimes. Such interval-monitored data are common in practice and have been studied carefully in the constant-stress and step-stress context \citep{gouno2001, balakrishnan2023exponential, balakrishnan2024gamma, balakrishnan2024weibull}, but no analogous treatment exists for CyALT models with lognormal lifetimes under interval monitoring.

    The second gap concerns robustness. Existing CyALT work has relied predominantly on the MLE. While it is asymptotically efficient under the assumed model, it is well known to be fragile in the presence of outliers or data contamination. In ALT experiments, where only a small number of units are tested under extreme conditions, even a handful of misrecorded or unusual observations can substantially distort parameter estimates and lead to unreliable reliability conclusions. The density power divergence (DPD) of \cite{basu1998robust} provides a principled robust alternative. Minimum density power divergence estimators (MDPDEs) reduce to the MLE when its tuning parameter is null. Moreover, MDPDEs become more and more robust as its tuning parameter increases, with only a modest and well-controlled loss in efficiency. DPD-based inference for interval-monitored step-stress ALTs has now been developed for the exponential ~ \citep{balakrishnan2023exponential}, lognormal \citep{balakrishnan2024lognormal}, gamma \citep{balakrishnan2024gamma}, and Weibull \citep{balakrishnan2024weibull} distributions, with \cite{balakrishnan2024competing} extending the framework to competing-risks settings. More recently, \cite{balakrishnan2025robust} developed robust divergence-based estimators for one-shot devices under cyclic ALTs. However, robust inference for CyALT models with lognormal lifetimes and interval monitoring remains unexplored, which is the focus of the present paper.

    The present paper tackles both of these gaps simultaneously by developing robust inferential procedures for CyALT models with lognormal lifetimes under interval monitoring. The main contributions are: (1) derivation of the likelihood function and the weighted MDPDE (WMDPDE) estimating equations for the interval-censored lognormal CyALT model; (2) establishment of the asymptotic distribution of the WMDPDE and the construction of asymptotic and bootstrap confidence intervals for the model parameters and key lifetime characteristics; (3) characterization of robustness through influence function analysis; and (4) a simulation study and an illustrative analysis of an automotive air-conditioner evaporator dataset  confirming the industry application advantages of the proposed approach.

    The paper is organized as follows. Section \ref{sec:industry_app} introduces the motivating example for an industry application. Section \ref{sec:model_description} describes the CyALT model and interval-monitoring setup under necessary life-testing assumptions. Section \ref{sec:WMDPDE} develops the WMDPDE and its asymptotic properties, including the confidence intervals. Section \ref{sec:lifetime_char_theory} presents important lifetime characteristics and their confidence intervals. Section \ref{sec:influence_func} derives the influence function. Section \ref{sec:sim_analysis} reports the simulation study. Section \ref{sec:real_analysis} analyses the evaporator dataset. Finally, the study is concluded in Section \ref{sec:conclusion}.

	\section{Motivation example}
	\label{sec:industry_app}
	This study is motivated by the analysis of an automotive air-conditioning system, which consists of four parts: a compressor, a condenser, an expansion valve, and an evaporator. The basic operating principle is the exchange of thermal energy between the condenser and the evaporator. The condenser, functioning as the heat emission component, converts high-pressure and high-temperature vapor into liquid through thermal exchange with air. In contrast, the evaporator performs the reverse role, and unlike the condenser, it is prone to freezing problems that reduce cooling efficiency. To prevent such freezing at the evaporator, a thermistor is used to continuously switch the compressor on and off. This repeated and cyclic operation of the compressor introduces mechanical stress, which may lead to cracks and leakage in the fins or tubes of the evaporator. 
	
	To aid understanding of this cyclic activity and the resulting need for a higher level of cyclic-stress loading in ALT, we provide a diagram in Figure \ref{fig:ac} illustrating the working principle of the air-conditioning system. This visual representation helps clarify how repeated compressor-to-evaporator cycles drive the motivation for applying cyclic-stress loading in reliability studies. Therefore, this study is particularly designed for any product that functions in a repetitive and cyclic manner. 

	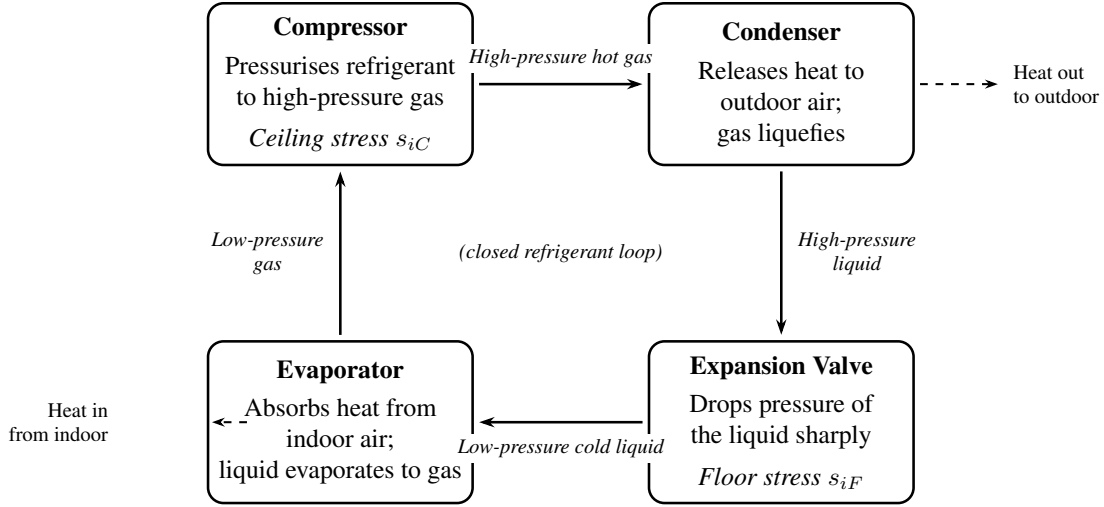
\begin{figure}[htbp]
		\centering
		\resizebox{0.9\textwidth}{!}{%
			\begin{tikzpicture}[
    mybox/.style={
        draw=black, fill=white, line width=1pt,
        rounded corners=6pt,
        minimum width=3.6cm, minimum height=2.2cm,
        align=center, font=\small
    },
    arr/.style={-{Stealth[length=6pt,width=4pt]},
                line width=1pt, color=black,
                shorten >=2pt, shorten <=2pt},
    darr/.style={-{Stealth[length=5pt,width=3.5pt]},
                 line width=0.8pt, dashed,
                 shorten >=2pt, shorten <=2pt},
    statelabel/.style={font=\scriptsize\itshape,
                       fill=white, inner sep=2pt, align=center}
]


\node[mybox] (comp) at (0,0) {
    \textbf{Compressor}\\[3pt]
    Pressurises refrigerant\\to high-pressure gas\\[4pt]
    \textit{Ceiling stress} $s_{iC}$
};

\node[mybox] (cond) at (6,0) {
    \textbf{Condenser}\\[3pt]
    Releases heat to\\outdoor air;\\gas liquefies
};

\node[mybox] (expv) at (6,-4.6) {
    \textbf{Expansion Valve}\\[3pt]
    Drops pressure of\\the liquid sharply\\[4pt]
    \textit{Floor stress} $s_{iF}$
};

\node[mybox] (evap) at (0,-4.6) {
    \textbf{Evaporator}\\[3pt]
    Absorbs heat from\\indoor air;\\liquid evaporates to gas
};


\draw[arr] (comp.east) -- (cond.west)
    node[statelabel, midway, above=4pt] {High-pressure hot gas};

\draw[arr] (cond.south) -- (expv.north)
    node[statelabel, midway, right=4pt] {High-pressure\\liquid};

\draw[arr] (expv.west) -- (evap.east)
    node[statelabel, midway, below=4pt] {Low-pressure cold liquid};

\draw[arr] (evap.north) -- (comp.south)
    node[statelabel, midway, left=4pt] {Low-pressure\\gas};


\draw[darr] (cond.east) -- ++(1.2,0)
    node[right, font=\scriptsize, align=left]
        {Heat out\\to outdoor};

\draw[darr] (-1.2,-4.6) -- (evap.west)
    node[left, font=\scriptsize, align=right, xshift=-1.2cm]
        {Heat in\\from indoor};

\node[font=\scriptsize\itshape] at (3,-2.3)
    {(closed refrigerant loop)};

\end{tikzpicture}   
		}
		\caption{Closed refrigerant cycle of an air-conditioner (AC).}
		\label{fig:ac}
	\end{figure}

    For reliability analysis, a cycle test is frequently employed to capture the degradation mechanisms of the air-conditioning system. Moreover, the operating profile and failure mechanisms of compressors in automotive systems are similar to those in compressors used in electronic appliances, making the methodology broadly applicable. This work focuses on cycle-based accelerated stress loading tests to obtain reliability information for such mechanical and electronic products. Cyclic-stress conditions $(s_i, i=1,\ldots,R)$ that exceed the normal use stress condition $(s_0)$ are used to accelerate failures, enabling timely reliability inference. For example, when $R = 2$, the setup involves three levels: use condition $(s_0)$, low cyclic-stress condition $(s_1)$, and high cyclic-stress condition $(s_2)$ as illustrated in Figure~\ref{Fig_2}. Under such setups, statistical inferences are important, as these inferences are further used in designing such CyALT experiments.     
    
	\section{Model Description}
	\label{sec:model_description}

    We now formalize the CyALT framework. Section~\ref{subsec:test_cond} introduces the important notations and test conditions, and Section~\ref{subsec:model_assump} states the lifetime and stress-life assumptions.
    
	\subsection{Test conditions under CyALT}
	\label{subsec:test_cond}
	Suppose that $K$ identical items are put on an experiment under CyALT and the items are exposed to $R$ number of different cyclic-stress conditions; $s_1, s_2, \dots, s_R$  that are higher than the use cyclic-stress condition. Each $s_i$ is characterized by its floor and ceiling levels. So, $s_{im}$ denotes the $m$-th level of the $i$-th stress condition $s_i$, where $m = F$ for the floor and $m = C$ for the ceiling level. ALTs are conducted at $s_i, i = 1,2, \ldots,R$. In air-conditioner's case, we assume that acceleration is implemented by only increasing the ceiling pressure where as the floor pressure level remains unchanged during the ALT. Therefore, we assume that $s_{1F} = s_{2F} = \ldots = s_{RF} < s_{1C} < s_{2C} < \ldots < s_{RC}$.  For each $i$-th stress condition, $s_{iC}$ is maintained for $(100\tau)\%$ of each cycle and then $s_{iF}$ is maintained for the rest of each cycle. The lifetimes of the units are measured in cycles and, at each cyclic-stress condition, the ALT continues until a pre-specified censoring time $t_c$ (observed in cycle). The total number of test units $(K)$ is given, and units are allocated to each cyclic-stress condition $(s_i, i = 1,2, \dots ,R)$ such that 
	\begin{align*}
		K_i= \pi_i K,  \text{ with }  0 < \pi_i < 1 \text{ and } \sum_{i=1}^R \pi_i =1 .
	\end{align*}
	We assume that the floor and ceiling levels of the use condition (i.e., $s_{0F}$ and $s_{0C}$) are known in advance. Also, the floor and ceiling levels in other accelerated stress conditions (i.e., $s_{1F} = s_{2F} = \dots = s_{RF}$, $s_{1C}, s_{2C}, \dots, s_{RC}$) are prespecified by the experimenter in application scenario and so as the ceiling stress maintaining time $\tau$ .
	
	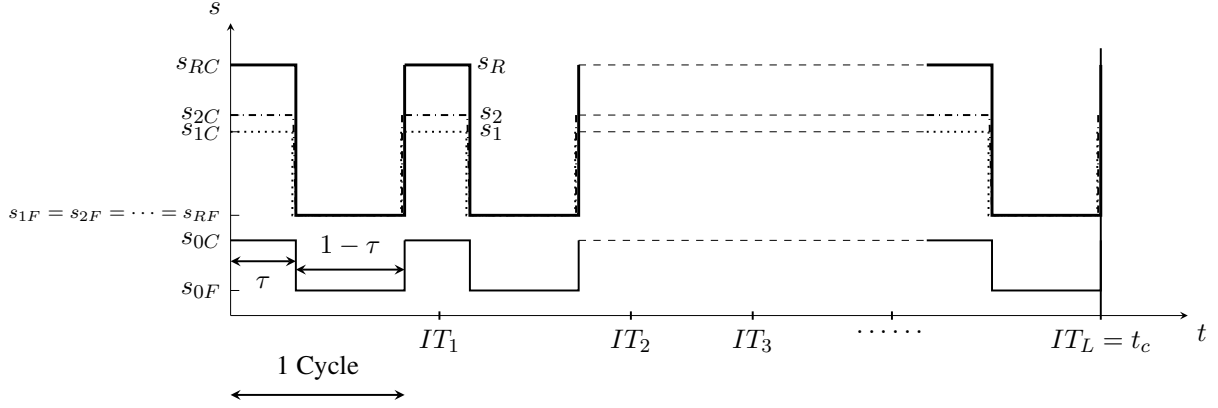
\begin{figure}[H]
		\centering
        \resizebox{0.98\textwidth}{!}{%
		\begin{tikzpicture}[x=2.5cm,y=1.2cm,>=stealth]
    \def\tauNum{0.375}   
			\def\tc{5}           
			\def\tauSym{\tau}
			\def\sZeroF{0.3}     
			\def\sZeroC{0.9}     
			\def\sF{1.2}         
			\def\sOneC{2.2}      
			\def\sTwoC{2.4}      
			\def\sRC{3.0}        
			
			\draw[->] (0,0) -- (\tc+0.5,0) node[below right] {$t$};
			\draw[->] (0,0) -- (0,3.5) node[above left] {$s$};
			
			\foreach \y/\lab in {
				\sZeroF/$s_{0F}$, 
				\sZeroC/$s_{0C}$, 
				\sF/\scriptsize{$s_{1F}=s_{2F}=\dots=s_{RF}$}, 
				\sOneC/$s_{1C}$, 
				\sTwoC/$s_{2C}$, 
				\sRC/$s_{RC}$
			} {
				\draw (0,\y) -- (0.05,\y);
				\node[left] at (0,\y) {\lab};
			}
			
			\foreach \k in {0,1} {
				\draw[dotted,thick] (\k,\sOneC) -- (\k+\tauNum-0.02,\sOneC)
				-- (\k+\tauNum-0.02,\sF)
				-- (\k+1-0.02,\sF)
				-- (\k+1-0.02,\sOneC);
			}
			
			\foreach \k in {0,1} {
				\draw[dash dot,thick] (\k,\sTwoC) -- (\k+\tauNum-0.015,\sTwoC)
				-- (\k+\tauNum,\sF)
				-- (\k+1-0.015,\sF)
				-- (\k+1-0.015,\sTwoC);
			}
			
			\foreach \k in {0,1} {
				\draw[very thick] (\k,\sRC) -- (\k+\tauNum,\sRC)
				-- (\k+\tauNum,\sF)
				-- (\k+1,\sF)
				-- (\k+1,\sRC);
			}
			
			\draw[dashed] (2,\sOneC)-- (\tc-1,\sOneC);
			\draw[dashed] (2,\sTwoC) -- (\tc-1,\sTwoC);
			\draw[dashed] (2,\sRC)   -- (\tc-1,\sRC);
			\draw[dashed] (2,\sZeroC)-- (\tc-1,\sZeroC);
			
			\foreach \k in {0,1} {
				\draw[thick] (\k,\sZeroC) -- (\k+\tauNum,\sZeroC)
				-- (\k+\tauNum,\sZeroF)
				-- (\k+1,\sZeroF)
				-- (\k+1,\sZeroC);
			}
			
			\foreach \x/\lbl in {1.2/$IT_1$,2.3/$IT_2$,3/$IT_3$, 3.8/$\dots\dots$} {
				\draw[thick,black] (\x,-0.05) -- (\x,0.05);
				\node[black,below=2pt] at (\x,0) {\lbl};
			}
			\draw[thick,black] (\tc,-0.05) -- (\tc,3.2);
			\node[black,below=2pt] at (\tc,0) {$IT_L=t_c$};
			
			\draw[<->, thick] (0,\sZeroC-0.25) -- (\tauNum,\sZeroC-0.25) 
			node[midway,below=2pt] {$\tauSym$};
			\draw[<->, thick] (\tauNum,\sZeroF+0.25) -- (1,\sZeroF+0.25) 
			node[midway,above=2pt] {$1-\tauSym$};
			\draw[<->, thick] (0,\sZeroF-1.25) -- (1,\sZeroF-1.25) 
			node[midway,above=2pt] {1 Cycle};
			
			\draw[dotted,thick] 
			(\tc-1,\sOneC) -- (\tc-1+\tauNum-0.02,\sOneC)
			-- (\tc-1+\tauNum-0.02,\sF)
			-- (\tc-0.02,\sF)
			-- (\tc-0.02,\sOneC);
			\draw[dash dot,thick] 
			(\tc-1,\sTwoC) -- (\tc-1+\tauNum-0.015,\sTwoC)
			-- (\tc-1+\tauNum,\sF)
			-- (\tc-0.015,\sF)
			-- (\tc-0.015,\sTwoC);
			\draw[very thick] 
			(\tc-1,\sRC) -- (\tc-1+\tauNum,\sRC)
			-- (\tc-1+\tauNum,\sF)
			-- (\tc,\sF)
			-- (\tc,\sRC);
			\draw[thick] 
			(\tc-1,\sZeroC) -- (\tc-1+\tauNum,\sZeroC)
			-- (\tc-1+\tauNum,\sZeroF)
			-- (\tc,\sZeroF)
			-- (\tc,\sZeroC);
			
			\node[] at (1.5,\sOneC) {$s_1$};
			\node[] at (1.5,\sTwoC) {$s_2$};
			\node[] at (1.5,\sRC) {$s_R$};
\end{tikzpicture}}
		\caption{Cyclic-stress ALT with $R$ accelerated cyclic-stress conditions: dotted = $s_1$ (floor $s_{1F}$, ceiling $s_{1C}$), thick dark = $s_2$ (floor $s_{2F}$, ceiling $s_{2C}$), dash-dotted = $s_R$ (floor $s_{RF}$, ceiling $s_{RC}$), and solid line = use stress (floor $s_{0F}$, ceiling $s_{0C}$).}
		\label{Fig_2}
	\end{figure}
	
	In addition, we consider $L$ pre-fixed inspection times (in cycles), $0 < IT_1  < IT_2  < \dots < IT_L = t_c$, where the number of failures at each inspection time under all the cyclic-stress conditions is recorded. Observed failures from intermittently-monitored items are grouped from all the R cyclic-stress conditions as a count of failures within each inspected interval. Let $n_{ij}$ be the number of failed items at the $i$-th stress condition during the $j$-th in interval of the inspection and $n_{i,L+1}$ represents the number of units exposed to $i$-th stress condition that survive until the censoring time $t_c$. (see Figure \ref{Fig_3}). Let $\mathcal{D}$ denote the observed count data, \textit{i.e.}, $ \mathcal{D} =\{ n_{ij}, i = 1, \ldots, R, j =1, \dots, L, n_{1,L+1}, \dots, n_{R,L+1}\}$. 
	
	\begin{figure}[H]
		\centering
		\vspace{1em}
        \resizebox{0.9\textwidth}{!}{%
		\begin{tikzpicture}[x=2.5cm,y=0.8cm,>=stealth]
			\def\tc{4.85}  
			\def\ITs{{0.5,1.1,2.1,3}}  
			\def\L{3}    
			
			\draw[->] (0.5,0) -- (\tc+0.4,0) node[below right] {$t$};
			
			\foreach \x/\lbl in {0.5/$0$,1.2/$IT_1$,2.3/$IT_2$,3/$IT_3$,3.8/$\dots\dots$,4.2/$IT_{L-1}$,4.85/$IT_L=t_c$} {
				\draw[thick,black] (\x,-0.05) -- (\x,0.05);
				\node[black,below=2pt] at (\x,0) {\lbl};
				
				\draw[thick,black] (\tc,-0.05) -- (\tc,5.2);
				
			}
			\foreach \i/\xstart/\xend in {1/0.5/1.2,2/1.2/2.3,3/2.3/3,4/4.2/4.85} {
				\draw[decorate,decoration={brace,amplitude=6pt}] 
				(\xstart,0.1) -- (\xend,0.1);
			}
			\node[above=1em] at (0.001*\tc/\L,0) {Stress Cond. 1};
			\node[above=1em] at ($ (0.5,0)!0.5!(1.2,0) $) {$n_{11}$};
			\node[above=1em] at ($ (1.2,0)!0.5!(2.3,0) $) {$n_{12}$};
			\node[above=1em] at ($ (2.3,0)!0.5!(3,0) $) {$n_{13}$};
			\node[above=1em] at ($ (4.2,0)!0.5!(4.85,0) $) {$n_{1L}$};
			\node[above=1em] at ($ (4.85,0)!0.5!(5.5,0) $) {$n_{1,L+1}$};
			
			\node[above=2.5em] at (0.001*\tc/\L,0) {Stress Cond. 2};
			\node[above=2.5em] at ($ (0.5,0)!0.5!(1.2,0) $) {$n_{21}$};
			\node[above=2.5em] at ($ (1.2,0)!0.5!(2.3,0) $) {$n_{22}$};
			\node[above=2.5em] at ($ (2.3,0)!0.5!(3,0) $) {$n_{23}$};
			\node[above=2.5em] at ($ (4.2,0)!0.5!(4.85,0) $) {$n_{2L}$};
			\node[above=2.5em] at ($ (4.85,0)!0.5!(5.5,0) $) {$n_{2,L+1}$};
			
			\node[above=5em] at (0.001*\tc/\L,0) {\vdots};
			
			\node[above=7.5em] at (0.001*\tc/\L,0) {Stress Cond. R};
			\node[above=7.5em] at ($ (0.5,0)!0.5!(1.2,0) $) {$n_{R1}$};
			\node[above=7.5em] at ($ (1.2,0)!0.5!(2.3,0) $) {$n_{R2}$};
			\node[above=7.5em] at ($ (2.3,0)!0.5!(3,0) $) {$n_{R3}$};
			\node[above=7.5em] at ($ (4.2,0)!0.5!(4.85,0) $) {$n_{RL}$};
			\node[above=7.5em] at ($ (4.85,0)!0.5!(5.5,0) $) {$n_{R,L+1}$};
\end{tikzpicture}
        }
		\caption{Observed data on the number of failures during an interval monitored a cyclic-stress accelerated life testing experiment.}
		\label{Fig_3}
	\end{figure}
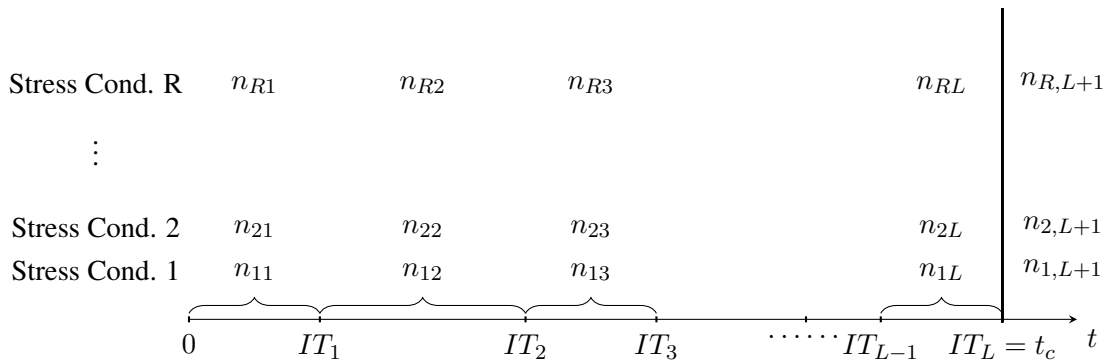
    
	\subsection{Model assumptions}
	\label{subsec:model_assump}
	We assume the following lifetime model in this paper: at any stress level $s_{im}$, that the lifetime of an unit follows a log-normal distribution with a scale parameter $\zeta(s_{im})$ and shape parameter $\sigma$. We assume that the lifetime characteristics of an item depends on the assigned stress value only through the scale parameter, whereas the shape parameter $\sigma$ is set free from the stress level. Then, its PDF is given by
	\begin{align*}
		f(x, \zeta(s_{im}), \sigma)	= \frac{1}{\sqrt{2\pi} \sigma x} \exp \left( -\frac{1}{2\sigma^2} \big\{ \ln (x) - \ln(\zeta(s_{im})) \big\}^2 \right). 
	\end{align*}  
	Furthermore, the relationship between stress and life is quantified through the scale parameter $\zeta(s_{im})$. Let $v_{im}$ denote the original physical stress variable (for example, voltage, temperature, or pressure) and let $g(~\cdot~)$ be a monotone function chosen according to the nature of the stress, for instance, $g(v) = v$ for a linear model, $g(v) = \ln (v)$ for the inverse power law model that is appropriate to mechanical stress, or $g(v) = 1/v$ for the Arrhenius model for thermal stress  \citep{nelson1990}. The scale parameter is then linked to $v_{im}$ through the log-linear relationship
	\begin{align}
		\ln( \zeta(v_{im})) = \gamma_0 + \gamma_1\, g(v_{im}),
		\label{eq:stress_life_original}
	\end{align}
	where $\gamma_0 \in \mathbb{R}$ and $\gamma_1 \in \mathbb{R}$, with the sign of $\gamma_1$ depending on whether $g(\cdot)$ is increasing or decreasing in the stress variable. To work on a common scale across different physical stress types, we define the standardized stress
	\begin{align}
		s_{im} = \frac{g(v_{im}) - g(v_{0F})}{g(v_{hC}) - g(v_{0F})},
		\label{eq:standardisation}
	\end{align}
	where $v_{0F}$ is the floor stress at the use condition and $v_{hC}$ is the ceiling stress at the highest accelerated condition. Under this standardization, equation \eqref{eq:stress_life_original} reduces to
	\begin{align}
		\ln( \zeta(s_{im}) )= \alpha_0 + \alpha_1 s_{im},
		\label{eq:stress_life}
	\end{align}
	where $\alpha_0 = \gamma_0 + \gamma_1\, g(v_{0F})$, $\alpha_1 = \gamma_1\,(g(v_{hC}) - g(v_{0F}))$, and $\alpha_1 < 0$ since higher standardized stress always corresponds to a shorter lifetime. Throughout this paper we work with equation \eqref{eq:stress_life} without loss of generality with $\alpha_0 \in \mathbb{R}$ and $\alpha_1 < 0$ as the unknown constants/parameters.
	The cumulative exposure (CE) model is assumed for the effect of changing stress levels. It should be noted that the scale parameter under the $i$-the stage of cyclic loading is a function of time. Under each of its cycles, the scale parameter is given by  
	\begin{align*}
		\zeta_i(x) = \begin{cases}
			\zeta(s_{iC}), &\text{ if} ~ 0<x\leq \tau\\
			\zeta(s_{iF}), &\text{ if} ~ \tau < x \leq 1.
		\end{cases}
	\end{align*}
	Therefore, cumulative exposure at $t$ number of cycles (see page 507, \cite{nelson2009accelerated}) is given by
	\begin{align*}
		\epsilon_i(t) = \int_0^t \frac{1}{\zeta_i(u)} du= t \int_0^1 \frac{1}{\zeta_i(u)} du = t \Big\{ \frac{\tau}{\zeta(s_{iC })} + \frac{1-\tau}{\zeta(s_{iF })}  \Big\}.
	\end{align*}
	The cumulative effect under $s_i$ cyclic-stress loading for a unit with $t$ cycles is given by
	\begin{align}
		\epsilon_i(t) =  t \Big\{ \frac{\tau}{\zeta(s_{iC })} + \frac{1-\tau}{\zeta(s_{iF })}  \Big\}.
	\end{align}  
	Therefore, the CDF at $t$ cycles for an item that is exposed to $i$-th stage of the cyclic-stress condition due to the cumulative exposure model is given by (page 508, \cite{nelson2009accelerated})
	\begin{align*}
		F_i(t) & = F(\epsilon_i(t), 1,  \sigma))  \\
		&	= \Phi\Bigg(\frac{\ln (\epsilon_i(t))}{\sigma}\Bigg),   \text{where~} \Phi ~ \text{is CDF of} ~ N(0,1).&\\
		F_i(t) &  = \Phi\Bigg(\frac{\ln (t) + \ln\big( \tau \eta_{iC }  + (1-\tau) \eta_{iF } \big) }{\sigma}\Bigg)
	\end{align*}
	with $\eta_{im } = 1/ \zeta(s_{im})= \exp \big\{-\alpha_0-\alpha_1 s_{im}\big\}, ~ \text{for} ~ m=C,F$.
	Hence, the PDF at $t$ cycles under $i$-th cyclic-stress loading is given by
	\begin{align*}
		f_i(t) & = \frac{1}{ \sigma t } \phi\Bigg(\frac{\ln (t) + \ln\big( \tau \eta_{iC }  + (1-\tau) \eta_{iF } \big) }{\sigma}\Bigg)
	\end{align*}
	where $\phi$ is the PDF of $N(0,1)$.
	
	Let $T_{ik}$ be the lifetime of the unit k under $s_i$. Assume $X_{ik} = \ln(T_{ik})$ as the log-lifetime random variable. Then clearly, 
	\begin{align}
		X_{ik} \sim N(\mu_i, \sigma)
	\end{align}
	with location $\mu_i= -\ln(\tau \eta_{iC } + (1-\tau) \eta_{iF })$ and standard deviation $\sigma$. Suppose that $T_{ik},~ i=1,2,\dots,R,~ k = 1,2,\dots, K_i$ is a random sample.  
	Note that the parameter vector of interest is $(\alpha_0, \alpha_1, \sigma)^\top$ with $\alpha_0 \in \mathbb{R}$, $\alpha_1 < 0$, and $\sigma > 0$, which we represent by $\boldsymbol{\theta} \in \boldsymbol{\Theta}$, where $\boldsymbol{\Theta}= \mathbb{R} \times \mathbb{R}^- \times \mathbb{R}^+$ denotes the parameter space. Furthermore, the probability of failure in the $j$-th interval under the $i$-th cyclic-stress condition is given by  
	\begin{align}
		\label{p_ij}
		p_{ij}(\boldsymbol{\theta}) = & \text{Pr} \left( IT_{j-1} < T_{ik} < IT_{j} \right) & \nonumber\\
		= & \text{Pr} \left( \ln(IT_{j-1}) < \ln(T_{ik}) < \ln(IT_{j}) \right)  & \nonumber\\
		= & \Phi\left( \frac{\ln(IT_{j})-\mu_i}{\sigma} \right) - \Phi \left( \frac{\ln(IT_{j-1})-\mu_i}{\sigma} \right), \quad i=1,2,\dots,R, \quad j = 1,2, \dots, L, &
	\end{align} and the probability of survival at the $i$-cyclic-stress condition is
	\begin{align}
		\label{p_i_l+1}
		p_{i,L+1}(\boldsymbol{\theta}) = 1 - \Phi\left( \frac{\ln(t_c)-\mu_i}{\sigma} \right), \quad i=1,2,\dots, R,
	\end{align}

    The probabilities $p_{ij}(\boldsymbol{\theta})$ and $p_{i,L+1}(\boldsymbol{\theta})$ defined in \eqref{p_ij} and \eqref{p_i_l+1} form the building blocks for all subsequent inference. In the next section, we use these to construct robust estimators of $\boldsymbol{\theta} = (\alpha_0, \alpha_1, \sigma)^\top$ based on the DPD.

	\section{Minimum Density Power Divergence Estimator}
	\label{sec:WMDPDE}
	
	In this section, we construct robust estimators for the CyALT model with 
	lognormal lifetimes using the DPD approach proposed 
	by \cite{basu1998robust}. Recall from Section \ref{sec:model_description} that under the $i$-th 
	cyclic-stress condition, the observed data consist of $n_{ij}$ failures in the 
	$j$-th inspection interval $(IT_{j-1}, IT_j)$, for $j = 1, 2, \ldots, L$, and 
	$n_{i,L+1}$ units that survive beyond the censoring time $t_c$. These $L+1$ 
	categories are mutually exclusive and exhaustive, so that $\sum_{j=1}^{L+1} n_{ij} = n_i$ for each $i = 1, 2, \ldots, R$.
	
	Under this setup, the observed counts at each stress condition follow a 
	multinomial distribution. Since the $R$ stress groups are independent, the 
	likelihood function is a product of $R$ multinomial likelihoods:
	\begin{align}
		\mathcal{L}(\boldsymbol{\theta}| \mathcal{D}) \propto 
		\prod_{i=1}^{R} \left( 
		\prod_{j=1}^{L} \big[ p_{ij}(\boldsymbol{\theta}) \big]^{n_{ij}} 
		\right) 
		\big[ p_{i,L+1}(\boldsymbol{\theta}) \big]^{n_{i,L+1}},
		\label{eq:likelihood}
	\end{align}
	where $p_{ij}(\boldsymbol{\theta})$ and $p_{i,L+1}(\boldsymbol{\theta})$ are as defined in 
	Section \ref{subsec:model_assump}. The log-likelihood function is given by 
	\begin{align}
		\ln\left( \mathcal{L}(\boldsymbol{\theta}| \mathcal{D}) \right) = \text{const} +  
		\sum_{i=1}^{R} \left( 	\sum_{j=1}^{L} n_{ij} \ln\left( p_{ij}(\boldsymbol{\theta}) \right) 
		+ n_{i,L+1} \ln\left( p_{i,L+1}(\boldsymbol{\theta}) \right) \right),
		\label{eq:log-likelihood}
	\end{align} and the maximum likelihood estimator (MLE) of $\boldsymbol{\theta}$ is given by
	\begin{align}
		\widehat{\boldsymbol{\theta}}_{\text{MLE}} = \argmax_{\boldsymbol{\theta} \in \boldsymbol{\Theta}}\; \ln \left(\mathcal{L}(\boldsymbol{\theta} | \mathcal{D}) \right).
		\label{eq:MLE}
	\end{align}
	Now, we are going to see the relation between the MLE and the minimization of the Kullback-Leibler divergence between two appropriate probability distributions. We define the theoretical probability vector for the group $i$ as
	\begin{align*}
		\mathbf{p}_i(\boldsymbol{\theta}) = \left( p_{i1}(\boldsymbol{\theta}),\, p_{i2}(\boldsymbol{\theta}),\, \ldots,\, 
		p_{iL}(\boldsymbol{\theta}),\, p_{i,L+1}(\boldsymbol{\theta}) \right)^\top, 
		\quad i = 1, 2, \ldots, R,
	\end{align*}
	and the corresponding empirical probability vector as
	\begin{align*}
		\widehat{\mathbf{p}}_i = \left( \frac{n_{i1}}{K_i},\, \frac{n_{i2}}{K_i},\, \ldots,\, 
		\frac{n_{iL}}{K_i},\, \frac{n_{i,L+1}}{K_i} \right)^\top, 
		\quad i = 1, 2, \ldots, R,
	\end{align*}
	where $K_i$ is the number of units assigned to the $i$-th stress condition. Note that $\sum_{j=1}^{L+1} p_{ij}(\boldsymbol{\theta}) = 1$ and $\sum_{j=1}^{L+1} n_{ij}/K_i = 1$ for each $i$, so both vectors lie in the same $(L+1)$-dimensional simplex, as required for the DPD. Here, $\sum_{j=1}^{L+1}$ includes the survival category for compactness, with $p_{ij}$ and $n_{ij}$ at $j = L+1$ denoting $p_{i,L+1}(\boldsymbol{\theta})$ and $n_{i,L+1}$ respectively. 
	
	The Kullback-Leibler divergence between the probability vectors $\widehat{\mathbf{p}}_i$ and $\mathbf{p}_i(\boldsymbol{\theta})$ is given by
	\begin{align}
		d_{\text{K-L}}(\widehat{\mathbf{p}}_i, \mathbf{p}_i(\boldsymbol{\theta})) 
		=  \sum_{j=1}^{L+1} \frac{n_{ij}}{K_i} 
		\ln\left( \frac{n_{ij}/K_i}{p_{ij}(\boldsymbol{\theta})} \right),
		\label{eq:KL}
	\end{align}
	For more detail about the Kullback-Leibler divergence measure, see \cite{pardo2006}. 
	
	Based on the Kullback-Leibler divergence, we shall consider the Weighted Kullback-Leibler Divergence (WKLD). The WKLD between the probability vectors $\widehat{\mathbf{p}}_i$ and $\mathbf{p}_i(\boldsymbol{\theta})$, $i=1,\dots, R$, is defined by
	\begin{align}
		\mathcal{W}(\boldsymbol{\theta}) 
		= \sum_{i=1}^{R} \frac{K_i}{K}  
		d_{\text{K-L}}(\widehat{\mathbf{p}}_i, \mathbf{p}_i(\boldsymbol{\theta})).
		\label{eq:WKLD}
	\end{align}
	This concept has been previously considered in \cite{balakrishnan2019robust, balakrishnan2020robust, balakrishnan2021robust, balakrishnan2023robust, balakrishnan2023competing}. Now, we are going to see the relation between the log-likelihood and the WKLD. we have
	\begin{align*}
		\mathcal{W}(\boldsymbol{\theta}) = & \frac{1}{K}  \sum_{i=1}^{R}   \sum_{j=1}^{L+1} n_{ij}
		\ln\left( \frac{n_{ij}/K_i}{p_{ij}(\boldsymbol{\theta})} \right)  & \\
		= & C -  \frac{1}{K}  \sum_{i=1}^{R}  \sum_{j=1}^{L+1} n_{ij}
		\ln\left( p_{ij}(\boldsymbol{\theta}) \right), &
	\end{align*}
	where $C = \frac{\text{const}}{K} + \frac{1}{K}  \sum\limits_{i=1}^{R}  \sum\limits_{j=1}^{L+1}  n_{ij} 
	\ln\left( n_{ij}/K_i\right) $ and it is a constant that does not depend on $\boldsymbol{\theta}$. Therefore, we have 
	\begin{align}
		\mathcal{W}(\boldsymbol{\theta}) = & C - \frac{1}{K}\ln\mathcal{L}(\boldsymbol{\theta};\mathcal{D}),&
		\label{eq:WKLD_loglik}
	\end{align}
	and then the MLE can be defined by
	\begin{align}
		\widehat{\boldsymbol{\theta}}_{\mathrm{MLE}} 
		= (\widehat{\alpha}_0, \widehat{\alpha}_1, \widehat{\sigma}) 
		= \argmin_{\boldsymbol{\theta} \in \boldsymbol{\Theta}}\, 
		\mathcal{W}(\boldsymbol{\theta}),
		\label{eq:MLE_WKLD}
	\end{align}
	where $\mathcal{W}(\boldsymbol{\theta})$ is the WKLD defined in~equation \eqref{eq:WKLD}. 
	
	Although the MLE defined in~equation \eqref{eq:MLE_WKLD} is asymptotically efficient under the assumed model, it is a BAN (Best Asymptotically Normal) estimator under the assumed model, it is well known to be sensitive to data contamination and model misspecification. Even a small proportion of outlying or misrecorded observations can lead to severely biased estimates and unreliable conclusions. This motivates the use of a robust alternative that retains high efficiency under clean data while providing protection against contamination. Based on the relation in equation \eqref{eq:WKLD_loglik}, several families of estimators have been proposed in the statistical literature by replacing the Kullback-Leibler divergence with a more robust divergence measure. In this paper we consider the DPD of \cite{basu1998robust} in order to obtain the MDPDE. 
	
	The DPD between the probability vectors $\widehat{\mathbf{p}}_i$ and $\mathbf{p}_i(\boldsymbol{\theta})$, for tuning parameter $\beta > 0$, is given below
	\begin{align}
		d_\beta(\widehat{\mathbf{p}}_i, \mathbf{p}_i(\boldsymbol{\theta})) 
		&= \sum_{j=1}^{L+1}\left\{
		p_{ij}(\boldsymbol{\theta})^{1+\beta} 
		- \left(1+\frac{1}{\beta}\right) \frac{n_{ij}}{K_i}
		p_{ij}(\boldsymbol{\theta})^{\beta}
		+ \frac{1}{\beta}
		\left(\frac{n_{ij}}{K_i}\right)^{1+\beta}
		\right\}. 
		\label{eq:DPD_full}
	\end{align}
	It is an easy exercise to see that
	\begin{align*}
		\lim_{\beta \downarrow 0}
		d_\beta(\widehat{\mathbf{p}}_i, \mathbf{p}_i(\boldsymbol{\theta})) 
		= d_{\text{K-L}}(\widehat{\mathbf{p}}_i, \mathbf{p}_i(\boldsymbol{\theta})).
	\end{align*}
	We can see that the term
	\begin{align*}
		\frac{1}{\beta}\sum_{j=1}^{L+1}
		\left(\frac{n_{ij}}{K_i}\right)^{1+\beta}
	\end{align*}
	does not depend on $\boldsymbol{\theta}$ and therefore has no importance in the minimization of $d_\beta(\widehat{\mathbf{p}}_i, \mathbf{p}_i(\boldsymbol{\theta}))$ 
	with respect to $\boldsymbol{\theta}$. Therefore, we define the weighted density power divergence as
	\begin{align}
		H_\beta(\boldsymbol{\theta}) = & \sum_{i=1}^{R} \frac{K_i}{K} d^*_\beta(\widehat{\mathbf{p}}_i, \mathbf{p}_i(\boldsymbol{\theta}))  &
		\label{eq:DPD}
	\end{align}
	where
	\begin{align}
		d^*_\beta(\widehat{\mathbf{p}}_i, \mathbf{p}_i(\boldsymbol{\theta})) 
		&= \sum_{j=1}^{L+1} p_{ij}(\boldsymbol{\theta})^{1+\beta}
		- \left(1+\frac{1}{\beta}\right)
		\sum_{j=1}^{L+1} \frac{n_{ij}}{K_i} p_{ij}(\boldsymbol{\theta})^{\beta}.
		\label{eq:DPD_reduced}
	\end{align}
	The weighted minimum density power divergence estimator (WMDPDE) with tuning parameter 
	$\beta$ is then defined as
	\begin{align}
		\widehat{\boldsymbol{\theta}}_\beta = (\widehat{\alpha}_{0,\beta}, \widehat{\alpha}_{1,\beta}, \widehat{\sigma}_\beta)  =\argmin_{\boldsymbol{\theta} \in \boldsymbol{\Theta}}\; H_\beta(\boldsymbol{\theta}),
		\label{eq:WMDPDE}
	\end{align} where
	\begin{align}
		H_\beta(\boldsymbol{\theta}) = & \frac{1}{K} \sum_{i=1}^{R} \left\{  
		\sum_{j=1}^{L+1} K_ip_{ij}(\boldsymbol{\theta})^{1+\beta}
		- \left(1+\frac{1}{\beta}\right) \sum_{j=1}^{L+1} n_{ij} p_{ij}(\boldsymbol{\theta})^{\beta} \right\}.
		\label{eq:DPD_full_reduced}
	\end{align}
	\begin{remark}
		As $\beta \to 0$, the DPD converges to the Kullback-Leibler divergence, and 
		$\widehat{\boldsymbol{\theta}}_\beta$ reduces to the MLE $\widehat{\boldsymbol{\theta}}_{\text{MLE}}$ 
		in~equation \eqref{eq:MLE}. 
	\end{remark}
	
	The WMDPDE $\widehat{\boldsymbol{\theta}}_\beta$ is obtained by minimising $H_\beta(\boldsymbol{\theta})$ 
	in~equation \eqref{eq:DPD_full_reduced} with respect to $\boldsymbol{\theta} = (\alpha_0, \alpha_1, \sigma)^\top$. 
	Setting the partial derivatives to zero gives the estimating equations
	\begin{align}
		\frac{\partial H_\beta(\boldsymbol{\theta})}{\partial \alpha_0} 
		&= (1+\beta) \sum_{i=1}^{R} \frac{K_i}{K} \left\{
		\sum_{j=1}^{L+1}  p_{ij}(\boldsymbol{\theta})^{\beta-1}
		\left[p_{ij}(\boldsymbol{\theta}) - \frac{n_{ij}}{K_i} \right]
		\frac{\partial p_{ij}(\boldsymbol{\theta})}{\partial \alpha_0} \right\}  = 0,& \nonumber\\
		\frac{\partial H_\beta(\boldsymbol{\theta})}{\partial \alpha_1} 
		&= (1+\beta) \sum_{i=1}^{R} \frac{K_i}{K}  \left\{
		\sum_{j=1}^{L+1} p_{ij}(\boldsymbol{\theta})^{\beta-1}
		\left[p_{ij}(\boldsymbol{\theta}) - \frac{n_{ij}}{K_i} \right]
		\frac{\partial p_{ij}(\boldsymbol{\theta})}{\partial \alpha_1} \right\} = 0, \nonumber\\
		\frac{\partial H_\beta(\boldsymbol{\theta})}{\partial \sigma} 
		&= (1+\beta) \sum_{i=1}^{R} \frac{K_i}{K}  \left\{
		\sum_{j=1}^{L+1} p_{ij}(\boldsymbol{\theta})^{\beta-1}
		\left[p_{ij}(\boldsymbol{\theta}) - \frac{n_{ij}}{K_i} \right]
		\frac{\partial p_{ij}(\boldsymbol{\theta})}{\partial \sigma} \right\} = 0.
		\label{eq:score}
	\end{align}
	To evaluate these equations, we require the partial derivatives of 
	$p_{ij}(\boldsymbol{\theta})$ and $p_{i,L+1}(\boldsymbol{\theta})$ with respect to $\alpha_0$, 
	$\alpha_1$, and $\sigma$. Recall from the equations \eqref{p_ij} and \eqref{p_i_l+1} in Subsection \ref{subsec:model_assump}, that 
	\begin{align*}
		p_{ij}(\boldsymbol{\theta}) =  \Phi\left( \frac{\ln(IT_{j})-\mu_i}{\sigma} \right) - \Phi \left( \frac{\ln(IT_{j-1})-\mu_i}{\sigma} \right),  ~ \text{and} ~ p_{i,L+1}(\boldsymbol{\theta}) = 1 - \Phi\left( \frac{\ln(t_c)-\mu_i}{\sigma} \right), 
	\end{align*} Therefore, to proceed further, let us define the standardized log-inspection time
	\begin{align}
		a_{i,j} = \frac{\ln(IT_j) - \mu_i}{\sigma}, 
		\quad j = 0, 1, \ldots, L,
		\label{eq:aij}
	\end{align}
	with $a_{i,0} = -\infty$ so that $\Phi(a_{i,0}) = 0$, and let $\phi(\cdot)$ 
	denote the standard normal PDF. Since 
	$\mu_i = -\ln(\tau\eta_{iC}+(1-\tau)\eta_{iF})$ and 
	$\eta_{im} = \exp\{-\alpha_0-\alpha_1 s_{im}\}$ for $m\in\{C,F\}$, 
	straightforward differentiation gives
	\begin{align*}
		\frac{\partial\mu_i}{\partial\alpha_0} = 1, \qquad
		\frac{\partial\mu_i}{\partial\alpha_1} = \bar{s}_i, 
	\end{align*}
	where 	
	\begin{align}
		\bar{s}_i = \frac{\tau s_{iC}\eta_{iC}+(1-\tau)s_{iF}\eta_{iF}}
		{\tau\eta_{iC}+(1-\tau)\eta_{iF}},
		\label{eq:sbar}
	\end{align}
	and consequently
	\begin{align*}
		\frac{\partial a_{i,j}}{\partial\alpha_0} = -\frac{1}{\sigma}, \qquad
		\frac{\partial a_{i,j}}{\partial\alpha_1} = -\frac{\bar{s}_i}{\sigma}, \qquad
		\frac{\partial a_{i,j}}{\partial\sigma}   = -\frac{a_{i,j}}{\sigma}.
	\end{align*}
	Differentiating $p_{ij}(\boldsymbol{\theta}) = \Phi(a_{i,j})-\Phi(a_{i,j-1})$ 
	for $j=1,\ldots,L$ and $p_{i,L+1}(\boldsymbol{\theta}) = 1-\Phi(a_{i,L})$ 
	with respect to $\alpha_0$, $\alpha_1$, and $\sigma$ gives
	\begin{align}
		\frac{\partial p_{ij}(\boldsymbol{\theta})}{\partial\alpha_0} 
		&= -\frac{1}{\sigma}\big[\phi(a_{i,j})-\phi(a_{i,j-1})\big],  &
		\frac{\partial p_{i,L+1}(\boldsymbol{\theta})}{\partial\alpha_0} 
		&= \phantom{-}\frac{\phi(a_{i,L})}{\sigma}, 
		\label{eq:dp_a0}\\
		\frac{\partial p_{ij}(\boldsymbol{\theta})}{\partial\alpha_1} 
		&= -\frac{\bar{s}_i}{\sigma}\big[\phi(a_{i,j})-\phi(a_{i,j-1})\big], &
		\frac{\partial p_{i,L+1}(\boldsymbol{\theta})}{\partial\alpha_1} 
		&= \phantom{-}\frac{\bar{s}_i\,\phi(a_{i,L})}{\sigma},
		\label{eq:dp_a1}\\
		\frac{\partial p_{ij}(\boldsymbol{\theta})}{\partial\sigma} 
		&= -\frac{1}{\sigma}\big[a_{i,j}\phi(a_{i,j})-a_{i,j-1}\phi(a_{i,j-1})\big], \quad &
		\frac{\partial p_{i,L+1}(\boldsymbol{\theta})}{\partial\sigma} 
		&= \phantom{-}\frac{a_{i,L}\,\phi(a_{i,L})}{\sigma}.
		\label{eq:dp_sigma}
	\end{align}
	Substituting~equation \eqref{eq:dp_a0}--equation \eqref{eq:dp_sigma} into~equation \eqref{eq:score}, 
	the three equations can be written compactly in matrix form as
	\begin{align}
		\sum_{i=1}^{R} \frac{K_i}{K}\,
		\mathbf{W}_i(\boldsymbol{\theta})^\top\,
		\mathbf{D}_i^{(\beta-1)}(\boldsymbol{\theta})\,
		\big[\mathbf{p}_i(\boldsymbol{\theta}) - \widehat{\mathbf{p}}_i\big] 
		= \mathbf{0}_3,
		\label{eq:estimeq}
	\end{align}
	where $\mathbf{0}_3$ is the $3$-dimensional null vector, $\mathbf{D}_i(\boldsymbol{\theta}) = 
	\mathrm{diag}(p_{i1}(\boldsymbol{\theta}),\ldots,p_{iL}(\boldsymbol{\theta}),p_{i,L+1}(\boldsymbol{\theta}))$ 
	is the $(L+1)\times(L+1)$ diagonal matrix of fitted probabilities, 
	$\mathbf{D}_i^{(r)}(\boldsymbol{\theta})$ denotes the diagonal matrix with entries 
	$p_{ij}(\boldsymbol{\theta})^{r}$, and $\mathbf{W}_i(\boldsymbol{\theta})$ is the 
	$(L+1)\times 3$ matrix whose $j$-th row is the gradient vector
	\begin{align*}
		\mathbf{w}_{ij}(\boldsymbol{\theta}) = \nabla_{\boldsymbol{\theta}} p_{ij}(\boldsymbol{\theta}) = 
		\left(\frac{\partial p_{ij}(\boldsymbol{\theta})}{\partial\alpha_0},\;
		\frac{\partial p_{ij}(\boldsymbol{\theta})}{\partial\alpha_1},\;
		\frac{\partial p_{ij}(\boldsymbol{\theta})}{\partial\sigma}\right)^\top,
		\quad j = 1,\ldots,L+1,
	\end{align*}
	as given explicitly in~equation \eqref{eq:dp_a0}-equation \eqref{eq:dp_sigma}, with the 
	convention that $p_{i,L+1}(\boldsymbol{\theta})$ corresponds to row $j = L+1$.

	
	\noindent The following theorem establishes the asymptotic distribution of 
	$\widehat{\boldsymbol{\theta}}_\beta$.
	
	\begin{theorem}\label{thm:asymp}
		Under the CyALT model of Section \ref{sec:model_description}, suppose that 
		$\mathbf{J}_\beta(\boldsymbol{\theta}_0)$ is non-singular and that the functions 
		$p_{i,j}(\boldsymbol{\theta})$ are twice continuously differentiable in a 
		neighbourhood of $\boldsymbol{\theta}_0$. Then $\widehat{\boldsymbol{\theta}}_\beta$ is a 
		consistent estimator of $\boldsymbol{\theta}_0 = (\alpha_{0,0}, \alpha_{1,0}, 
		\sigma_0)^\top$, the true value of $\boldsymbol{\theta}$, and
		\begin{align}
			\sqrt{K} \left( \widehat{\boldsymbol{\theta}}_\beta - \boldsymbol{\theta}_0 \right)
			\;\xrightarrow{\mathcal{L}}\;
			N_3\!\Big(\mathbf{0}_3,\;
			\mathbf{J}_\beta(\boldsymbol{\theta}_0)^{-1}\,
			\mathbf{K}_\beta(\boldsymbol{\theta}_0)\,
			\mathbf{J}_\beta(\boldsymbol{\theta}_0)^{-1}
			\Big),
			\quad \text{as } K \to \infty,
			\label{eq:asymp}
		\end{align}
		where
		\begin{align}
			\mathbf{J}_\beta(\boldsymbol{\theta}) &= \sum_{i=1}^{R} \frac{K_i}{K}
			\mathbf{W}_i(\boldsymbol{\theta})^\top\,
			\mathbf{D}_i^{(\beta-1)}(\boldsymbol{\theta})\,
			\mathbf{W}_i(\boldsymbol{\theta}),
			\label{eq:Jbeta}\\[4pt]
			\mathbf{K}_\beta(\boldsymbol{\theta}) &= \sum_{i=1}^{R} \frac{K_i}{K}\,
			\mathbf{W}_i(\boldsymbol{\theta})^\top
			\Big[
			\mathbf{D}_i^{(2\beta-1)}(\boldsymbol{\theta})
			- \mathbf{D}_i^{(\beta)}(\boldsymbol{\theta})\,\mathbf{1}_{L+1}\,
			\mathbf{1}_{L+1}^\top\,\mathbf{D}_i^{(\beta)}(\boldsymbol{\theta})
			\Big]
			\mathbf{W}_i(\boldsymbol{\theta}),
			\label{eq:Kbeta}
		\end{align}
		with $\mathbf{D}_i^{(r)}(\boldsymbol{\theta})$ denoting the $(L+1)\times(L+1)$ 
		diagonal matrix with $j$-th diagonal entry $p_{i,j}(\boldsymbol{\theta})^r$, and 
		$\mathbf{1}_{L+1}$ the $(L+1)$-dimensional vector of ones.
	\end{theorem}
	
	\begin{proof}
		The asymptotic distribution follows from the general theory of \cite{ghosh2013robust} for minimum density power divergence 	estimators under independent but not identically distributed observations. In the CyALT setting, the $R$ stress conditions 	constitute independent groups with sample sizes $K_1, \ldots, K_R$ (satisfying $\sum_{i=1}^R K_i= K$). At each stress level $i$, the data consist of failure counts $(n_{i1}, \ldots, n_{i,L+1})$ following a multinomial distribution with $K_i$ trials and a probability vector $\mathbf{p}_i(\boldsymbol{\theta}) = (p_{i1}(\boldsymbol{\theta}), \ldots, p_{i,L+1} (\boldsymbol{\theta}))^\top$ of dimension $L+1$. From \cite{ghosh2013robust}, the asymptotic distribution of the WMDPDE is given by 
		\begin{align}
			\sqrt{K} \left( \widehat{\boldsymbol{\theta}}_\beta - \boldsymbol{\theta}_0 \right) 
			\xrightarrow{\mathcal{L}} N_3 \left( \mathbf{0}_3, \mathbf{J}_\beta(\boldsymbol{\theta}_0)^{-1} \mathbf{K}_\beta(\boldsymbol{\theta}_0)  \mathbf{J}_\beta(\boldsymbol{\theta}_0)^{-1}  \right), \quad \text{as } K \to \infty,
		\end{align} 
		where $\mathbf{J}_\beta(\boldsymbol{\theta})$ and $\mathbf{K}_\beta(\boldsymbol{\theta})$ are defined in terms of the log-probability gradients.  
		
		Specifically, let $\mathbf{u}_{ij}(\boldsymbol{\theta}) = \nabla_{\boldsymbol{\theta}} \ln p_{ij}(\boldsymbol{\theta})$ denote the gradient of the log-probability. Then 
		\begin{align}
			\mathbf{J}_\beta(\boldsymbol{\theta}) &= \sum_{i=1}^{R} \sum_{j=1}^{L+1} \frac{K_i}{K}
			\mathbf{u}_{ij}(\boldsymbol{\theta}) 
			\mathbf{u}_{ij}(\boldsymbol{\theta})^\top 
			p_{ij}(\boldsymbol{\theta})^{\beta+1}, \\
			\mathbf{K}_\beta(\boldsymbol{\theta}) &= \mathbf{J}_{2\beta}(\boldsymbol{\theta}) 
			- \sum_{i=1}^{R} \frac{K_i}{K} \boldsymbol{\xi}_{i,\beta}(\boldsymbol{\theta}) \boldsymbol{\xi}_{i,\beta}(\boldsymbol{\theta})^\top,
		\end{align} 
		where
		\begin{align}
			\boldsymbol{\xi}_{i,\beta}(\boldsymbol{\theta}) &= \sum_{j=1}^{L+1} \mathbf{u}_{ij}(\boldsymbol{\theta}) p_{ij}(\boldsymbol{\theta})^{\beta+1}.
		\end{align}
		
		\noindent 
		Since $\mathbf{u}_{ij}(\boldsymbol{\theta}) = \frac{1}{p_{ij}(\boldsymbol{\theta})} \nabla_{\boldsymbol{\theta}} p_{ij}(\boldsymbol{\theta}) = \frac{\mathbf{w}_{ij}(\boldsymbol{\theta})}{p_{ij}(\boldsymbol{\theta})}$, we have
		\begin{align}
			\mathbf{u}_{ij}(\boldsymbol{\theta}) \mathbf{u}_{ij}(\boldsymbol{\theta})^\top p_{ij}(\boldsymbol{\theta})^{\beta+1} 
			&= \frac{\mathbf{w}_{ij}(\boldsymbol{\theta})}{p_{ij}(\boldsymbol{\theta})} \frac{\mathbf{w}_{ij}(\boldsymbol{\theta})^\top}{p_{ij}(\boldsymbol{\theta})} p_{ij}(\boldsymbol{\theta})^{\beta+1} \nonumber\\
			&= \mathbf{w}_{ij}(\boldsymbol{\theta}) \mathbf{w}_{ij}(\boldsymbol{\theta})^\top p_{ij}(\boldsymbol{\theta})^{\beta-1}.
		\end{align}
		Therefore,
		\begin{align}
			\mathbf{J}_\beta(\boldsymbol{\theta}) &= \sum_{i=1}^{R} \sum_{j=1}^{L+1} \frac{K_i}{K}
			\mathbf{w}_{ij}(\boldsymbol{\theta}) \mathbf{w}_{ij}(\boldsymbol{\theta})^\top p_{ij}(\boldsymbol{\theta})^{\beta-1} \nonumber \\
			&= \sum_{i=1}^{R} \frac{K_i}{K}
			\mathbf{W}_i(\boldsymbol{\theta})^\top \mathbf{D}_i^{(\beta-1)}(\boldsymbol{\theta}) \mathbf{W}_i(\boldsymbol{\theta}),
		\end{align}
		where $\mathbf{W}_i(\boldsymbol{\theta})$ is the $(L+1) \times 3$ matrix with rows $\mathbf{w}_{ij}(\boldsymbol{\theta})$, and $\mathbf{D}_i^{(\beta-1)}(\boldsymbol{\theta}) = \mathrm{diag}(p_{i1}(\boldsymbol{\theta})^{\beta-1}, \ldots, p_{i,L+1}(\boldsymbol{\theta})^{\beta-1})$.
		
		Similarly, for $\boldsymbol{\xi}_{i,\beta}(\boldsymbol{\theta})$,
		\begin{align}
			\boldsymbol{\xi}_{i,\beta}(\boldsymbol{\theta}) &= \sum_{j=1}^{L+1} \mathbf{u}_{ij}(\boldsymbol{\theta}) p_{ij}(\boldsymbol{\theta})^{\beta+1} \nonumber\\
			&= \sum_{j=1}^{L+1} \frac{\mathbf{w}_{ij}(\boldsymbol{\theta})}{p_{ij}(\boldsymbol{\theta})} p_{ij}(\boldsymbol{\theta})^{\beta+1} \nonumber\\
			&= \sum_{j=1}^{L+1} \mathbf{w}_{ij}(\boldsymbol{\theta}) p_{ij}(\boldsymbol{\theta})^{\beta} \nonumber \\
			&= \mathbf{W}_i(\boldsymbol{\theta})^\top \mathbf{D}_i^{(\beta)}(\boldsymbol{\theta}) \mathbf{1}_{L+1},
		\end{align}
		where $\mathbf{D}_i^{(\beta)}(\boldsymbol{\theta}) = \mathrm{diag}(p_{i1}(\boldsymbol{\theta})^{\beta}, \ldots, p_{i,L+1}(\boldsymbol{\theta})^{\beta})$ and $\mathbf{1}_{L+1}$ is the $(L+1)$-dimensional vector of ones.
		
		Thus,
		\begin{align}
			\boldsymbol{\xi}_{i,\beta}(\boldsymbol{\theta}) \boldsymbol{\xi}_{i,\beta}(\boldsymbol{\theta})^\top 
			&= \mathbf{W}_i(\boldsymbol{\theta})^\top \mathbf{D}_i^{(\beta)}(\boldsymbol{\theta}) \mathbf{1}_{L+1} \mathbf{1}_{L+1}^\top \mathbf{D}_i^{(\beta)}(\boldsymbol{\theta}) \mathbf{W}_i(\boldsymbol{\theta}).
		\end{align}
		
		Therefore,
		\begin{align}
			\mathbf{K}_\beta(\boldsymbol{\theta}) &= \sum_{i=1}^{R} \frac{K_i}{K} \mathbf{W}_i(\boldsymbol{\theta})^\top \mathbf{D}_i^{(2\beta-1)}(\boldsymbol{\theta}) \mathbf{W}_i(\boldsymbol{\theta}) \nonumber\\
			&\quad - \sum_{i=1}^{R} \frac{K_i}{K} \mathbf{W}_i(\boldsymbol{\theta})^\top \mathbf{D}_i^{(\beta)}(\boldsymbol{\theta}) \mathbf{1}_{L+1} \mathbf{1}_{L+1}^\top \mathbf{D}_i^{(\beta)}(\boldsymbol{\theta}) \mathbf{W}_i(\boldsymbol{\theta}) \nonumber\\
			&= \sum_{i=1}^{R} \frac{K_i}{K} \mathbf{W}_i(\boldsymbol{\theta})^\top \Big[ \mathbf{D}_i^{(2\beta-1)}(\boldsymbol{\theta}) - \mathbf{D}_i^{(\beta)}(\boldsymbol{\theta}) \mathbf{1}_{L+1} \mathbf{1}_{L+1}^\top \mathbf{D}_i^{(\beta)}(\boldsymbol{\theta}) \Big] \mathbf{W}_i(\boldsymbol{\theta}),
		\end{align}
		which completes the derivation of the asymptotic covariance matrix in the statement of the theorem.
		
	\end{proof}
	
	
	\noindent Since $\boldsymbol{\theta}_0$ is unknown, the covariance matrix 
	in~equation \eqref{eq:asymp} is estimated by replacing $\boldsymbol{\theta}_0$ with 
	$\widehat{\boldsymbol{\theta}}_\beta$:
	\begin{align}
		\widehat{\mathbf{\Sigma}}_\beta = \frac{1}{K}
		\mathbf{J}_\beta(\widehat{\boldsymbol{\theta}}_\beta)^{-1}\,
		\mathbf{K}_\beta(\widehat{\boldsymbol{\theta}}_\beta)\,
		\mathbf{J}_\beta(\widehat{\boldsymbol{\theta}}_\beta)^{-1},
		\label{eq:Sigma_hat}
	\end{align}
	which is consistent for 
	$\mathbf{J}_\beta(\boldsymbol{\theta}_0)^{-1}\mathbf{K}_\beta(\boldsymbol{\theta}_0)
	\mathbf{J}_\beta(\boldsymbol{\theta}_0)^{-1}$ by the continuous mapping theorem. 
	
	\noindent Denoting the components of $\widehat{\boldsymbol{\theta}}_\beta$ 
	corresponding to $\alpha_0$, $\alpha_1$, and $\sigma$ by 
	$\widehat{\alpha}_{0,\beta}$, $\widehat{\alpha}_{1,\beta}$, and 
	$\widehat{\sigma}_\beta$ respectively, and the diagonal entries of 
	$\widehat{\mathbf{\Sigma}}_\beta$ by 
	$\widehat{\sigma}^2_{\alpha_0}$, $\widehat{\sigma}^2_{\alpha_1}$, 
	and $\widehat{\sigma}^2_{\sigma}$, approximate $100(1-\gamma)\%$ 
	confidence intervals for $\alpha_{0,0}$, $\alpha_{1,0}$, and 
	$\sigma_0$ are
	\begin{align}
		\widehat{\alpha}_{0,\beta} \pm z_{\gamma/2}\,
		\widehat{\sigma}_{\alpha_0},
		\qquad
		\widehat{\alpha}_{1,\beta} \pm z_{\gamma/2}\,
		\widehat{\sigma}_{\alpha_1},
		\qquad
		\widehat{\sigma}_\beta \pm z_{\gamma/2}\,
		\widehat{\sigma}_{\sigma},
		\label{eq:CI}
	\end{align}
	where $z_{\gamma/2}$ is the upper $\gamma/2$ quantile of the 
	standard normal distribution.

    While the confidence intervals in \eqref{eq:CI} establishes inference on the model parameters directly, practitioners are often more interested in lifetime characteristics at the use condition. Some important characteristics of interest are median lifetime, mean time to failure, and reliability at a pre-specified mission time. These are nonlinear functions of $\boldsymbol{\theta}$, and their estimation and inference are developed in the following section.
	
	\section{Lifetime Characteristics at the Use Condition}
    \label{sec:lifetime_char_theory}
	\subsection{Their Definitions and Point Estimates}
	\label{subsubsec:characteristics_pointestimates}
	At the use condition $(s_{0F}, s_{0C})$ with cyclic fraction $\tau$, the location parameter of the lognormal lifetime distribution is
	\begin{equation}
		\mu_0 = \mu_0(\boldsymbol{\theta})
		= -\ln\!\left[
		\tau\,e^{-\alpha_0 - \alpha_1 s_{0C}}
		+ (1-\tau)\,e^{-\alpha_0 - \alpha_1 s_{0F}}
		\right],
		\label{eq:S_mu0}
	\end{equation}
	so that $\log T \sim N(\mu_0, \sigma^2)$ at the use condition. 
	
	\noindent\textbf{Median.}
	The $q$-th quantile of the lifetime distribution at the use condition is 
	\begin{equation}
		t_{q,0} = \exp \bigl(\mu_0 + \sigma\,\Phi^{-1}(q)\bigr),
		\qquad q \in (0,1).
		\label{eq:S_quantile}
	\end{equation}
	Since $\Phi^{-1}(0.5) = 0$, the median simplifies to $t_{0.5,0} = e^{\mu_0}$.
	
	\noindent\textbf{Mean time to failure.}
	For $T \sim \mathrm{Lognormal}(\mu_0, \sigma^2)$,
	\begin{equation}
		\mathrm{MTTF}_0 = E[T] = E[e^X] =
		\exp \left(\mu_0 + \frac{\sigma^2}{2}\right),
		\label{eq:S_mttf}
	\end{equation}
	where $X \sim N(\mu_0, \sigma^2)$. Since $\sigma > 0$, the MTTF always exceeds the median: $\mathrm{MTTF}_0 > t_{0.5,0}$ for all $\sigma > 0$.
	
	\noindent\textbf{Reliability function.}
	The reliability at mission time $t_0 > 0$ is
	\begin{equation}
		R_0(t_0) = P(T > t_0)
		= 1 - \Phi\!\left(\frac{\ln t_0 - \mu_0}{\sigma}\right)
		= 1 - \Phi(a_0),
		\label{eq:S_reliability}
	\end{equation}
	where $a_0 = (\ln t_0 - \mu_0)/\sigma$. Now it is straight forward to obtain the estimates of the above characteristics by plugging-in the WMDPDE s,   $\widehat{\alpha}_{0,\beta}$, $\widehat{\alpha}_{1,\beta}$ and $\widehat{\sigma}_{\beta}$ in the equations \eqref{eq:S_quantile}, \eqref{eq:S_mttf} and \eqref{eq:S_reliability}
	
	\subsection{Confidence Intervals}
	\label{subsec:CIs} 
    
	\subsubsection{Direct and Transformed Asymptotic Confidence Intervals}
	
	Note that all the above lifetime characteristics are a smooth scalar function of the parameter $\boldsymbol{\theta}$. We first obtain the asymptotic results on the confidence intervals for an arbitrary smooth function $\psi(\boldsymbol{\theta})$, then it may be implemented analogously on all the lifetime characteristics.  By Theorem \ref{thm:asymp}, $\sqrt{K}\,(\widehat{\boldsymbol{\theta}}_\beta - \boldsymbol{\theta}_0)$ converges in distribution to  $N(\boldsymbol{0},\,\boldsymbol{\Sigma}_\beta(\boldsymbol{\theta}_0))$, where $\boldsymbol{\Sigma}_\beta (\boldsymbol{\theta})= \boldsymbol{J}_\beta^{-1} (\boldsymbol{\theta}) \boldsymbol{K}_\beta (\boldsymbol{\theta}) \boldsymbol{J}_\beta^{-1}(\boldsymbol{\theta})$. Consequently, by applying the Delta method, we have
	\begin{equation}
		\sqrt{K}\,\bigl(\psi(\widehat{\boldsymbol{\theta}}_\beta)
		- \psi(\boldsymbol{\theta}_0)\bigr)
		\xrightarrow[K\to\infty]{\mathcal{L}}
		N\!\left(0,\;
		\boldsymbol{d}_\psi^\top\,
		\boldsymbol{\Sigma}_\beta\,
		\boldsymbol{d}_\psi\right),
		\label{eq:S_asym_psi}
	\end{equation}
	where $\boldsymbol{d}_\psi =
	\partial\psi/\partial\boldsymbol{\theta}|_{\boldsymbol{\theta}_0}$ is the vector of partial derivatives of $\psi$ with respect to $(\alpha_0, \alpha_1, \sigma)$. The asymptotic standard error of $\psi(\widehat{\boldsymbol{\theta}}_\beta)$ is therefore estimated as
	\begin{equation}
		\widehat{\mathrm{SE}}(\psi) =
		\sqrt{\,\boldsymbol{d}_\psi^\top\,
			\widehat{\boldsymbol{\Sigma}}_\beta\,
			\boldsymbol{d}_\psi\,},
		\label{eq:S_se_psi}
	\end{equation}
	with $\boldsymbol{d}_\psi$ evaluated at $\widehat{\boldsymbol{\theta}}_\beta$ and $\widehat{\boldsymbol{\Sigma}}_\beta$ from equation \eqref{eq:Sigma_hat}. Thus, the $100(1-\gamma)\%$ asymptotically approximate confidence interval for $\psi(\boldsymbol{\theta})$ is then given by 
	\[ \text{IC}_\gamma(\psi(\boldsymbol{\theta})) = \left[ \widehat{\psi} \pm z_{\gamma/2}, \widehat{\mathrm{SE}}(\psi) \right].
	\] 
	Defining
	\begin{equation*}
		\bar{s}_0 =
		\frac{\tau\,s_{0C}\,e^{-\alpha_0-\alpha_1 s_{0C}}
			+ (1-\tau)\,s_{0F}\,e^{-\alpha_0-\alpha_1 s_{0F}}}
		{\tau\,e^{-\alpha_0-\alpha_1 s_{0C}}
			+ (1-\tau)\,e^{-\alpha_0-\alpha_1 s_{0F}}},
	\end{equation*}
	and by straightforward differentiation, the vectors of the partial derivatives of $\psi$ are as follows
	\begin{equation}
		\boldsymbol{d}_{t_{q,0}}
		= t_{q,0}
		\begin{pmatrix} 1 \\ \bar{s}_0 \\ \Phi^{-1}(q)
		\end{pmatrix}, \qquad
		\boldsymbol{d}_{\mathrm{MTTF}_0}
		= \mathrm{MTTF}_0
		\begin{pmatrix} 1 \\ \bar{s}_0 \\ \sigma
		\end{pmatrix}, \qquad
		\boldsymbol{d}_{R_0}
		= \frac{\phi(a_0)}{\sigma}
		\begin{pmatrix} 1 \\ \bar{s}_0 \\ a_0
		\end{pmatrix}.
		\label{eq:S_dvecs}
	\end{equation}
	Thus, the $100(1-\gamma)\%$ asymptotically approximate confidence interval for the lifetime characteristics are as follows
	\begin{align}
		\text{IC}_\gamma(t_{0.5,0}(\boldsymbol{\theta})) & = \left[ \widehat{t}_{0.5,0} \pm z_{\gamma/2}, \widehat{\mathrm{SE}}(t_{0.5,0}) \right],\\ \text{IC}_\gamma(\mathrm{MTTF}_0(\boldsymbol{\theta})) & = \left[ \widehat{\mathrm{MTTF}_0} \pm z_{\gamma/2}, \widehat{\mathrm{SE}}(\mathrm{MTTF}_0) \right], \\  
		\text{IC}_\gamma(R_0(\boldsymbol{\theta}))& = \left[ \widehat{R_0} \pm z_{\gamma/2}, \widehat{\mathrm{SE}}(R_0) \right],
		\label{eq:S_CIs}
	\end{align}
	where the point estimates $\widehat{t}_{0.5,0}$, $\widehat{\mathrm{MTTF}_0}$ and $\widehat{R_0}$ are obtained in \ref{subsubsec:characteristics_pointestimates} and the standard errors $\widehat{\mathrm{SE}}(t_{0.5,0})$, $\widehat{\mathrm{SE}}(\mathrm{MTTF}_0)$, and $\widehat{\mathrm{SE}}(R_0)$ can be obtained from the equations \eqref{eq:S_se_psi} and \eqref{eq:S_dvecs}. These direct CIs are asymptotically valid but do not respect natural parameter constraints, as the lower bound may be negative for lifetimes or outside $(0,1)$ for reliability at moderate sample sizes. Applying the same argument to a monotone transformation $\psi$ before constructing the confidence interval, then inverting back, resolves this without altering the asymptotic validity. Thus, we now present the transformed confidence intervals for these characteristics.
	
	\noindent\textbf{Log-transformed CI for median and MTTF:}
	
	\noindent For $\psi > 0$, applying equation \eqref{eq:S_asym_psi} to $\ln\psi$ gives $\mathrm{SE}(\ln\widehat{\psi}) \approx \widehat{\mathrm{SE}}(\psi)/\widehat{\psi}$.
	Constructing a confidence interval on the log scale and exponentiating, the $100(1-\gamma)\%$ log-transformed confidence intervals for the median and MTTF are
	\begin{align}
		\text{IC}_\gamma(t_{0.5,0}) &=
		\left[
		\widehat{t}_{0.5,0}\cdot\exp\!\left(
		-z_{\gamma/2}\,\frac{\widehat{\mathrm{SE}}(t_{0.5,0})}
		{\widehat{t}_{0.5,0}}\right),\;
		\widehat{t}_{0.5,0}\cdot\exp\!\left(
		+z_{\gamma/2}\,\frac{\widehat{\mathrm{SE}}(t_{0.5,0})}
		{\widehat{t}_{0.5,0}}\right)
		\right],
		\label{eq:S_ci_log_med} \\[6pt]
		\text{IC}_\gamma(\mathrm{MTTF}_0) &=
		\left[
		\widehat{\mathrm{MTTF}}_0\cdot\exp\!\left(
		-z_{\gamma/2}\,\frac{\widehat{\mathrm{SE}}(\mathrm{MTTF}_0)}
		{\widehat{\mathrm{MTTF}}_0}\right),\;
		\widehat{\mathrm{MTTF}}_0\cdot\exp\!\left(
		+z_{\gamma/2}\,\frac{\widehat{\mathrm{SE}}(\mathrm{MTTF}_0)}
		{\widehat{\mathrm{MTTF}}_0}\right)
		\right],
		\label{eq:S_ci_log_mttf}
	\end{align}
	both strictly positive and asymmetric around the point estimate, better reflecting the skewed sampling distribution of lognormal lifetime estimators.

	\noindent\textbf{Logit-transformed CI for reliability:}
	
	\noindent For $R_0(t_0)\in(0,1)$, applying equation \eqref{eq:S_asym_psi} to $\mathrm{logit}(\psi) = \ln(\psi/(1-\psi))$ gives $\mathrm{SE}(\mathrm{logit}\,\hat{\psi}) \approx
	\widehat{\mathrm{SE}}(\psi)/(\hat{\psi}(1-\hat{\psi}))$.
	Setting
	\begin{equation*}
		S = \exp\!\left(\frac{z_{\gamma/2}\,
			\widehat{\mathrm{SE}}(R_0)}
		{\widehat{R}_0\,(1-\widehat{R}_0)}\right)
	\end{equation*}
	and inverting back using
	$\mathrm{logit}^{-1}(y) = e^y/(1+e^y)$, the $100(1-\gamma)\%$ logit-transformed confidence interval for the reliability is
	\begin{equation}
		\text{IC}_\gamma(R_0) =
		\left(
		\frac{\widehat{R}_0}{\widehat{R}_0 + (1-\widehat{R}_0)\,S},\;
		\frac{\widehat{R}_0}{\widehat{R}_0 + (1-\widehat{R}_0)/S}
		\right) \subset (0,1),
		\label{eq:S_ci_logit}
	\end{equation}
	with both bounds guaranteed to lie in $(0,1)$.
	
	\subsubsection{BCa bootstrap confidence interval}
	
	As an alternative to nonparametric and asymptotic CIs and suitable to small-sample cyclic-stress ALT experiments, we report the parametric bias-corrected and accelerated (BCa) percentile bootstrap confidence interval of \cite{efron1994introduction}. For a given $\beta$, suppose we have $B$ parametric bootstrap samples generated from the WMDPDE ,  $\widehat{\boldsymbol{\theta}}_\beta$ which is obtained using the observed data. let $\widehat{\psi}^{(1)}, \ldots, \widehat{\psi}^{(B)}$ be the corresponding bootstrap estimates of $\psi$ and let $\widehat{\psi}_{(1)} \leq \cdots \leq \widehat{\psi}_{(B)}$ denote their order statistics. The two-sided $100(1-\gamma)\%$ BCa confidence interval is
	\begin{equation}
		\text{IC}_{1-\gamma}(\psi) =
		\Bigl[
		\widehat{\psi}_{([q_1 B])},\;
		\widehat{\psi}_{([q_2 B])}
		\Bigr],
		\label{eq:S_ci_bca}
	\end{equation}
	where $[\,\cdot\,]$ denotes the integer part and
	\begin{equation}
		q_1 = \Phi\!\left(\widehat{z}_0 +
		\frac{\widehat{z}_0 - z_{\gamma/2}}
		{1 - \widehat{a}\,(\widehat{z}_0 - z_{\gamma/2})}\right),
		\qquad
		q_2 = \Phi\!\left(\widehat{z}_0 +
		\frac{\widehat{z}_0 + z_{\gamma/2}}
		{1 - \widehat{a}\,(\widehat{z}_0 + z_{\gamma/2})}\right).
		\label{eq:S_bca_q}
	\end{equation}
	The bias-correction constant $\widehat{z}_0$ and acceleration constant $\widehat{a}$ are estimated as follows. The bias-correction is estimated by
	\begin{equation}
		\widehat{z}_0 = \Phi^{-1}\!\left(
		\frac{\#\bigl\{\widehat{\psi}^{(b)} \leq
			\widehat{\psi}\bigr\}}{B}
		\right),
		\label{eq:S_z0}
	\end{equation}
	where $\widehat{\psi}$ is the original estimate based on the observed data. Following \cite{efron1994introduction}, a suggested estimate of the acceleration is given by
	\begin{equation}
		\widehat{a} = \frac{1}{6}
		\left[\,\sum_{l=1}^{n_{RL}}
		\Bigl(\widehat{\psi}^{[-l]} -
		\widehat{\psi}^{(\cdot)}\Bigr)^2
		\right]^{-3/2}
		\left[\,\sum_{l=1}^{n_{RL}}
		\Bigl(\widehat{\psi}^{[-l]} -
		\widehat{\psi}^{(\cdot)}\Bigr)^3
		\right],
		\label{eq:S_accel}
	\end{equation}
	where $\widehat{\psi}^{[-l]}$ is the WMDPDE of $\psi$ computed
	from the observed data with the $l$-th failure removed
	(the jackknife estimate), $n_{RL} = \sum_{i=1}^{R}
	\sum_{j=1}^{L} n_{ij}$ is the total number of observed
	failures across all stress groups and inspection intervals,
	and
	\begin{equation*}
		\widehat{\psi}^{(\cdot)} =
		\frac{1}{n_{RL}}\sum_{l=1}^{n_{RL}} \widehat{\psi}^{[-l]}.
	\end{equation*}

    Having established estimation and inference procedures for the model parameters and lifetime characteristics, we now move to demonstrate the theoretical robustness properties of the WMDPDE  through influence function analysis, which characterizes the effect of data contamination on the estimator.
    
	\section{Influence function for the WMDPDE }
	\label{sec:influence_func}
    
	Let us denote $F_{\boldsymbol{\theta }}^{l}$ for the assumed distribution of 
	$l$-th stress condition with a mass function $\mathbf{p}_{l}\left( 
	\boldsymbol{\theta }\right) =( p_{l1}\left( \boldsymbol{\theta }\right)
	, $ $\dots, p_{l,L+1}\left( \boldsymbol{\theta }\right) ) ^{T},$ $l=1,\dots,R,$
	under the cyclic-stress ALT and interval monitoring with log-normal
	lifetimes, and we also denote $\boldsymbol{F}_{\boldsymbol{\theta }}=\left(
	F_{\boldsymbol{\theta }}^{1},\dots,F_{\boldsymbol{\theta }}^{R}\right) .$ Let $%
	G^{l}$ denote the true distribution underlying the data with mass function $%
	\boldsymbol{g}^{l}=\left( g_{l1},\dots,g_{l,L+1}\right)^{T}, l=1,\dots,R$ and we
	write, 
	\begin{equation*}
		\boldsymbol{G=}\left( G^{1},\dots,G^{R}\right) \text{ and }\boldsymbol{g}%
		=\left( \boldsymbol{g}^{1},\dots,\boldsymbol{g}^{R}\right) .
	\end{equation*}%
	The minimum density power divergence statistical functional,$\boldsymbol{\ T}%
	_{\beta }\left( G^{1},\dots,G^{R}\right) ,$ is defined as the minimizer of the
	WDPD for the vectors 
	\begin{equation*}
		\mathbf{p}_{l}\left( \boldsymbol{\theta }\right) =\left( p_{l1}\left( 
		\boldsymbol{\theta }\right) ,\dots,p_{l,L+1}\left( \boldsymbol{\theta }\right)
		\right) ^{T}\text{ and }\boldsymbol{g}^{l}=\left( g_{l1},\dots,g_{l,L+1}\right)
		^{T},
	\end{equation*}%
	$l=1,\dots,R,$ i.e.$,$%
	\begin{equation*}
		\tsum\limits_{l=1}^{R}\frac{K_{l}}{K}d_{\beta }\left( \boldsymbol{g}^{l},%
		\boldsymbol{T}_{\beta }\left( G^{1},\dots,G^{R}\right) \right) =\min_{%
			\boldsymbol{\theta }\text{ }\in \text{ }\Theta }\tsum\limits_{l=1}^{R}\frac{%
			K_{l}}{K}d_{\beta }\left( \boldsymbol{g}^{l},\mathbf{p}_{l}\left( 
		\boldsymbol{\theta }\right) \right) ,
	\end{equation*}%
	Let $\boldsymbol{\theta }^{T}\boldsymbol{=T}_{\beta }\left(
	G^{1},\dots,G^{R}\right) $ be the minimum weighted density power divergence
	functional with contamination  in all $R$  stress. Consider a contaminate
	version of \ the true lifetime distribution $G^{l}$ by 
	\begin{equation*}
		G_{\varepsilon }^{l}=(1-\varepsilon )G^{l}+\varepsilon \Delta _{t_{0}^{l}}
	\end{equation*}%
	with $\varepsilon $ the contamination proportion and $\Delta _{t_{0}^{l}}$ being the degenerate distribution at the contamination point $t_{0}^{l}$. In the discretized model under consideration, contamination is represented as a cell contamination in the $l$-th stress condition, and so the contamination for $t_{0}^{l}$ should have
	all elements equal to zero except for only one component.
	
	We represent by 
	\begin{equation*}
		\boldsymbol{g}_{\varepsilon }^{l}=(1-\varepsilon )\boldsymbol{g}%
		^{l}+\varepsilon \boldsymbol{e}_{t_{0}^{l}}
	\end{equation*}%
	the resulting multinomial vector associated to $G^{l}$ In which $\boldsymbol{%
		e}_{t_{0}^{l}}$ has all elements zero except for the interval containing the
	point perturbation $t_{0}^{l}.$
	
	Let  $\boldsymbol{\theta }_{l,\varepsilon }^{T}=\boldsymbol{T}_{\beta
	}\left( G^{1},\dots,G^{l-1},G_{\varepsilon }^{l},G^{l+1},\dots,G^{R}\right) $ be the
	minimum weighted density power divergence functional with contamination only
	in the $l$-th stress condition, the influence function (IF) is defined by%
	\begin{eqnarray*}
		\text{IF}(t_{0}^{l},\boldsymbol{T}_{\beta },G^{1},\dots,G^{R})) &=&\lim_{\varepsilon
			\rightarrow 0}\frac{\boldsymbol{T}_{\beta }\left(
			G^{1},\dots,G^{l-1},G_{\varepsilon }^{l},G^{l+1},\dots,G^{R}\right) -\boldsymbol{T}%
			_{\beta }(G^{1},\dots,G^{R})}{\varepsilon } \\
		&=&\left( \frac{\partial \boldsymbol{T}_{\beta }\left(
			G^{1},\dots,G^{l-1},G_{\varepsilon }^{l},G^{l+1},\dots,G^{R}\right) }{\partial
			\varepsilon }\right) _{\varepsilon =0} \\
		&=&\left( \frac{\partial \boldsymbol{\theta }_{l,\varepsilon }^{T}}{\partial
			\varepsilon }\right) _{\varepsilon =0}
	\end{eqnarray*}%
	and the IF in all the $R$ stress conditions is given by%
	\begin{eqnarray*}
		\text{IF}(t_{0}^{1},\dots,t_{0}^{R},\boldsymbol{T}_{\beta },G^{1},\dots,G^{R}))
		&=&\lim_{\varepsilon \rightarrow 0}\frac{\boldsymbol{T}_{\beta }\left(
			G_{\varepsilon }^{1},\dots,G_{\varepsilon }^{l-1},G_{\varepsilon
			}^{l},G_{\varepsilon }^{l+1},\dots,G_{\varepsilon }^{R}\right) -\boldsymbol{T}%
			_{\beta }\left( G^{1},\dots,G^{R}\right) }{\varepsilon } \\
		&=&\left( \frac{\partial \boldsymbol{T}_{\beta }\left( G_{\varepsilon
			}^{1},\dots,G_{\varepsilon }^{l-1},G_{\varepsilon }^{l},G_{\varepsilon
			}^{l+1},\dots,G_{\varepsilon }^{R}\right) }{\partial \varepsilon }\right)
		_{\varepsilon =0} \\
		&=&\left( \frac{\partial \boldsymbol{\theta }_{\varepsilon }^{T}}{\partial
			\varepsilon }\right) _{\varepsilon =0}.
	\end{eqnarray*}%
	In the next Theorem we shall get the IF in two previous situations.
	
	\begin{theorem}
		The IF of the WMDPDE of the cyclic-stress ALT and interval monitoring with
		log-normal lifetime at $l$-th stress condition at a point contamination $%
		t_{0}^{l}$ and the assumed distribution $\boldsymbol{F}_{\boldsymbol{\theta }%
			_{0}}=\left( F_{\boldsymbol{\theta }_{0}}^{1},\dots,F_{\boldsymbol{\theta }%
			_{0}}^{R}\right) $ is given by%
		\begin{equation*}
			\text{IF} \left( t_{0}^{l},\boldsymbol{T}_{\beta },\boldsymbol{F}_{\boldsymbol{%
					\theta }_{0}}\right) =\boldsymbol{J}_{\beta }\left( \boldsymbol{\theta }%
			_{0}\right) ^{-1}\frac{K_{l}}{K}\tsum \limits_{j=1}^{L+1}p_{lj}\left( 
			\boldsymbol{\theta }_{0}\right) ^{\beta -1}\left( \frac{\partial
				p_{lj}\left( \boldsymbol{\theta }\right) }{\partial \boldsymbol{\theta }}%
			\right) _{\boldsymbol{\theta }=\boldsymbol{\theta }_{0}}\left( -p_{lj}\left( 
			\boldsymbol{\theta }_{0}\right) +\boldsymbol{e}_{t_{0}^{l}}\right)
		\end{equation*}%
		and in the case of all $R$ stress conditions at a point contamination $%
		\left( t_{0}^{1},\dots,t_{0}^{R}\right) $ by%
		\begin{equation*}
			\text{IF}\left( t_{0}^{1},\dots,t_{0}^{R},\boldsymbol{T}_{\beta },\boldsymbol{F}_{%
				\boldsymbol{\theta }_{0}}\right) =\boldsymbol{J}_{\beta }\left( \boldsymbol{%
				\theta }_{0}\right) ^{-1}\tsum \limits_{l=1}^{R}\frac{K_{l}}{K}\tsum
			\limits_{j=1}^{L+1}p_{lj}\left( \boldsymbol{\theta }_{0}\right) ^{\beta
				-1}\left( \frac{\partial p_{lj}\left( \boldsymbol{\theta }\right) }{\partial 
				\boldsymbol{\theta }}\right) _{\boldsymbol{\theta }=\boldsymbol{\theta }%
				_{0}}\left( -p_{lj}\left( \boldsymbol{\theta }_{0}\right) +\boldsymbol{e}%
			_{t_{0}^{l}}\right) .
		\end{equation*}
	\end{theorem}
	
	\begin{proof}
		First we are going to get the IF\ when we only have contamination in the $l$-th stress condition. We shall denote 
		\begin{equation*}
			F_{\boldsymbol{\theta }_{l},\varepsilon }^{l}=\left( 1-\varepsilon \right)
			F_{\boldsymbol{\theta }_{0}}^{l}+\varepsilon \boldsymbol{e}_{t_{0}^{l}}
		\end{equation*}%
		and by $\mathbf{p}_{l,\varepsilon }\left( \boldsymbol{\theta }%
		_{l,\varepsilon }\right) =\left( 1-\varepsilon \right) \mathbf{p}_{l}\left( 
		\boldsymbol{\theta }_{0}\right) +\varepsilon \Delta _{t_{0}^{l}}$ the
		corresponding mass function and\\ $\boldsymbol{\theta }_{l,\varepsilon }^{T}=%
		\boldsymbol{T}_{\beta }\left( F_{\boldsymbol{\theta }_{0}}^{1},\dots,F_{%
			\boldsymbol{\theta }_{0}}^{l-1},F_{\boldsymbol{\theta }_{l},\varepsilon
		}^{l},F_{\boldsymbol{\theta }_{0}}^{l+1},\dots,F_{\boldsymbol{\theta }%
			_{0}}^{R}\right) .$ In this situation 
		\begin{equation*}
			G_{\varepsilon }^{l}=(1-\varepsilon )F_{\boldsymbol{\theta }%
				_{0}}^{l}+\varepsilon \Delta _{t_{0}^{l}}.
		\end{equation*}%
		The minimum weighted DPD functional, $\boldsymbol{\theta }_{l,\varepsilon
		}^{T},$ must satisfy the estimating equations, 
		\begin{eqnarray*}
			&&\tsum\limits_{\substack{ i=1 \\ i\neq l}}^{R}\frac{K_{i}}{K}%
			\tsum\limits_{j=1}^{L+1}p_{ij}\left( \boldsymbol{\theta }_{l,\varepsilon
			}\right) ^{\beta -1}\left( p_{ij}\left( \boldsymbol{\theta }_{l,\varepsilon
			}\right) -\boldsymbol{g}_{j}^{i}\right) \left( \frac{\partial p_{ij}\left( 
				\boldsymbol{\theta }\right) }{\partial \boldsymbol{\theta }}\right) _{%
				\boldsymbol{\theta }=\boldsymbol{\theta }_{l,\varepsilon }} \\
			&&+\frac{K_{l}}{K}\tsum\limits_{j=1}^{L+1}p_{lj}\left( \boldsymbol{\theta }%
			_{l,\varepsilon }\right) ^{\beta -1}\left( p_{lj}\left( \boldsymbol{\theta }%
			_{l,\varepsilon }\right) -\boldsymbol{g}_{\varepsilon ,j}^{l}\right) \left( 
			\frac{\partial p_{lj}\left( \boldsymbol{\theta }\right) }{\partial 
				\boldsymbol{\theta }}\right) _{\boldsymbol{\theta }=\boldsymbol{\theta }%
				_{l,\varepsilon }} \\
			&=&\boldsymbol{0}_{3}.
		\end{eqnarray*}%
		Impliciting differentiating, with respect to $\varepsilon ,$ in the previous
		estimating equations, we get
	
	\begin{eqnarray*}
		&&\tsum \limits_{\substack{ i=1  \\ i\neq l}}^{R}\frac{K_{i}}{K}\tsum
		\limits_{j=1}^{L+1}\left( \beta -1\right) p_{ij}\left( \boldsymbol{\theta }%
		_{l,\varepsilon }\right) ^{\beta -2}\left( \frac{\partial p_{ij}\left( 
			\boldsymbol{\theta }\right) }{\partial \boldsymbol{\theta }^{T}}\right) _{%
			\boldsymbol{\theta =\theta }_{l,\varepsilon }}\frac{\partial \boldsymbol{\theta 
			}_{l,\varepsilon }^{T}}{\partial \varepsilon }\left( p_{ij}\left( \boldsymbol{%
			\theta }_{l,\varepsilon }\right) -\boldsymbol{g}_{j}^{i}\right) \left( \frac{%
			\partial p_{ij}\left( \boldsymbol{\theta }\right) }{\partial \boldsymbol{%
				\theta }}\right) _{\boldsymbol{\theta }=\boldsymbol{\theta }_{l,\varepsilon }}
		\\
		&&+\text{ }\tsum \limits_{\substack{ i=1  \\ i\neq l}}^{R}\frac{K_{i}}{K}%
		\tsum \limits_{j=1}^{L+1}p_{ij}\left( \boldsymbol{\theta }_{l,\varepsilon
		}\right) ^{\beta -1}\left( \left( \frac{\partial p_{ij}\left( \boldsymbol{%
				\theta }\right) }{\partial \boldsymbol{\theta }^{T}}\right) _{\boldsymbol{%
				\theta =\theta }_{l,\varepsilon }}\frac{\partial \boldsymbol{\theta }%
			_{l,\varepsilon }^{T}}{\partial \varepsilon }-\frac{\partial \boldsymbol{g}%
			_{j}^{i}}{\partial \varepsilon }\right) \left( \frac{\partial p_{ij}\left( 
			\boldsymbol{\theta }\right) }{\partial \boldsymbol{\theta }}\right) _{%
			\boldsymbol{\theta }=\boldsymbol{\theta }_{l,\varepsilon }} \\
		&&+\tsum \limits_{\substack{ i=1  \\ i\neq l}}^{R}\frac{K_{i}}{K}\tsum
		\limits_{j=1}^{L+1}p_{ij}\left( \boldsymbol{\theta }_{l,\varepsilon }\right)
		^{\beta -1}\left( p_{ij}\left( \boldsymbol{\theta }_{l,\varepsilon }\right) -%
		\boldsymbol{g}_{\varepsilon }^{i}\right) \left( \frac{\partial
			^{2}p_{ij}\left( \boldsymbol{\theta }\right) }{\partial \boldsymbol{\theta }%
			\text{ }\boldsymbol{\partial \theta }^{T}}\right) _{\boldsymbol{\theta }=%
			\boldsymbol{\theta }_{l,\varepsilon }}\frac{\partial \boldsymbol{\theta }%
			_{l,\varepsilon }^{T}}{\partial \varepsilon } \\
		&&+\frac{K_{l}}{K}\tsum \limits_{j=1}^{L+1}\left( \beta -1\right)
		p_{lj}\left( \boldsymbol{\theta }_{l,\varepsilon }\right) ^{\beta -2}\left( 
		\frac{\partial p_{lj}\left( \boldsymbol{\theta }\right) }{\partial 
			\boldsymbol{\theta }^{T}}\right) _{\boldsymbol{\theta =\theta }_{l,\varepsilon
		}}\frac{\partial \boldsymbol{\theta }_{l,\varepsilon }^{T}}{\partial
			\varepsilon }\left( p_{lj}\left( \boldsymbol{\theta }_{l,\varepsilon }\right) -%
		\boldsymbol{g}_{\varepsilon ,j}^{l}\right) \left( \frac{\partial
			p_{lj}\left( \boldsymbol{\theta }\right) }{\partial \boldsymbol{\theta }}%
		\right) _{\boldsymbol{\theta }=\boldsymbol{\theta }_{l,\varepsilon }} \\
		&&+\frac{K_{l}}{K}\tsum \limits_{j=1}^{L+1}p_{lj}\left( \boldsymbol{\theta }%
		_{l,\varepsilon }\right) ^{\beta -1}\left( \left( \frac{\partial p_{lj}\left( 
			\boldsymbol{\theta }\right) }{\partial \boldsymbol{\theta }^{T}}\right) _{%
			\boldsymbol{\theta =\theta }_{l,\varepsilon }}\frac{\partial \boldsymbol{\theta 
			}_{l,\varepsilon }^{T}}{\partial \varepsilon }-\frac{\partial \boldsymbol{g}%
			_{\varepsilon ,j}^{l}}{\partial \varepsilon }\right) \left( \frac{\partial
			p_{lj}\left( \boldsymbol{\theta }\right) }{\partial \boldsymbol{\theta }}%
		\right) _{\boldsymbol{\theta }=\boldsymbol{\theta }_{l,\varepsilon }} 
    \end{eqnarray*} %
    \begin{eqnarray*}
		&&+\frac{K_{l}}{K}\tsum \limits_{j=1}^{L+1}p_{lj}\left( \boldsymbol{\theta }%
		_{l,\varepsilon }\right) ^{\beta -1}\left( p_{lj}\left( \boldsymbol{\theta }%
		_{l,\varepsilon }\right) -\boldsymbol{g}_{\varepsilon ,j}^{l}\right) \left( 
		\frac{\partial ^{2}p_{lj}\left( \boldsymbol{\theta }\right) }{\partial 
			\boldsymbol{\theta }\text{ }\boldsymbol{\partial \theta }^{T}}\right) _{%
			\boldsymbol{\theta }=\boldsymbol{\theta }_{l,\varepsilon }}\frac{\partial 
			\boldsymbol{\theta }_{l,\varepsilon }^{T}}{\partial \varepsilon } \\
		&=&\boldsymbol{0}_{3}.
	\end{eqnarray*}%
	For $\varepsilon =0$ we have, $p_{lj}\left( \boldsymbol{\theta }_{l,\varepsilon
	}\right) =p_{lj}\left( \boldsymbol{\theta }_{0}\right) $, 
	\begin{eqnarray*}
		\boldsymbol{g}_{\varepsilon ,j}^{l} &=&(1-\varepsilon )\boldsymbol{g}%
		_{j}^{l}+\varepsilon \boldsymbol{e}_{t_{0}^{l}}\underset{\varepsilon =0}{=}%
		(1-\varepsilon )p_{lj}\left( \boldsymbol{\theta }_{0}\right) +\varepsilon 
		\boldsymbol{e}_{t_{0}^{l}} \\
		&=&p_{lj}\left( \boldsymbol{\theta }_{0}\right) .
	\end{eqnarray*}%
	and 
	\begin{equation*}
		\frac{\partial \boldsymbol{g}_{\varepsilon ,j}^{i}}{\partial \varepsilon }%
		=\left \{ 
		\begin{array}{ll}
			0 & i\neq l \\ 
			-p_{lj}\left( \boldsymbol{\theta }_{0}\right) +\boldsymbol{e}_{t_{0}^{l}} & 
			i=l%
		\end{array}%
		\right. .
	\end{equation*}
	
	Therefore for $\varepsilon =0$ the previous equation system is given by. 
	\begin{eqnarray*}
		&&\text{ }\tsum\limits_{\substack{ i=1 \\ i\neq l}}^{R}\frac{K_{i}}{K}%
		\tsum\limits_{j=1}^{L+1}p_{ij}\left( \boldsymbol{\theta }_{0}\right) ^{\beta
			-1}\left( \frac{\partial p_{ij}\left( \boldsymbol{\theta }\right) }{\partial 
			\boldsymbol{\theta }^{T}}\right) _{\boldsymbol{\theta =\theta }_{0}}\left( 
		\frac{\partial p_{ij}\left( \boldsymbol{\theta }\right) }{\partial 
			\boldsymbol{\theta }}\right) _{\boldsymbol{\theta }=\boldsymbol{\theta }%
			_{0}}\text{IF}\left( t_{0}^{l},\boldsymbol{T}_{\beta },\boldsymbol{F}_{\boldsymbol{%
				\theta }_{0}}\right)  \\
		&&+\frac{K_{l}}{K}\tsum\limits_{j=1}^{L+1}p_{lj}\left( \boldsymbol{\theta }%
		_{0}\right) ^{\beta -1}\left( \left( \frac{\partial p_{lj}\left( \boldsymbol{%
				\theta }\right) }{\partial \boldsymbol{\theta }^{T}}\right) _{\boldsymbol{%
				\theta =\theta }_{0}}\text{IF}\left( t_{0}^{l},\boldsymbol{T}_{\beta },\boldsymbol{F%
		}_{\boldsymbol{\theta }_{0}}\right) -\frac{\partial \boldsymbol{g}%
			_{\varepsilon ,j}^{l}}{\partial \varepsilon }\right) \left( \frac{\partial
			p_{lj}\left( \boldsymbol{\theta }\right) }{\partial \boldsymbol{\theta }}%
		\right) _{\boldsymbol{\theta }=\boldsymbol{\theta }_{0}} \\
		&=&\boldsymbol{0}_{3}
	\end{eqnarray*}%
	or equivalently 
	\begin{eqnarray*}
		&&\text{ }\tsum\limits_{\substack{ i=1 \\ i\neq l}}^{R}\frac{K_{i}}{K}%
		\tsum\limits_{j=1}^{L+1}p_{ij}\left( \boldsymbol{\theta }_{0}\right) ^{\beta
			-1}\left( \frac{\partial p_{ij}\left( \boldsymbol{\theta }\right) }{\partial 
			\boldsymbol{\theta }^{T}}\right) _{\boldsymbol{\theta =\theta }_{0}}\left( 
		\frac{\partial p_{ij}\left( \boldsymbol{\theta }\right) }{\partial 
			\boldsymbol{\theta }}\right) _{\boldsymbol{\theta }=\boldsymbol{\theta }%
			_{0}}\text{IF}\left( t_{0}^{l},\boldsymbol{T}_{\beta },\boldsymbol{F}_{\boldsymbol{%
				\theta }_{0}}\right)  \\
		&&+\frac{K_{l}}{K}\tsum\limits_{j=1}^{L+1}p_{lj}\left( \boldsymbol{\theta }%
		_{0}\right) ^{\beta -1}\left( \frac{\partial p_{lj}\left( \boldsymbol{\theta 
			}\right) }{\partial \boldsymbol{\theta }^{T}}\right) _{\boldsymbol{\theta
				=\theta }_{0}}\left( \frac{\partial p_{lj}\left( \boldsymbol{\theta }\right) 
		}{\partial \boldsymbol{\theta }}\right) _{\boldsymbol{\theta }=\boldsymbol{%
				\theta }_{0}}\text{IF}\left( t_{0}^{l},\boldsymbol{T}_{\beta },\boldsymbol{F}_{%
			\boldsymbol{\theta }_{0}}\right)  \\
		&&-\frac{K_{l}}{K}\tsum\limits_{j=1}^{L+1}p_{lj}\left( \boldsymbol{\theta }%
		_{0}\right) ^{\beta -1}\left( -p_{lj}\left( \boldsymbol{\theta }_{0}\right) +%
		\boldsymbol{e}_{t_{0}^{l}}\right) \left( \frac{\partial p_{lj}\left( 
			\boldsymbol{\theta }\right) }{\partial \boldsymbol{\theta }}\right) _{%
			\boldsymbol{\theta }=\boldsymbol{\theta }_{0}} \\
		&=&\boldsymbol{0}_{3}.
	\end{eqnarray*}%
	Therefore we have, 
	\begin{equation*}
		\boldsymbol{J}_{\beta }\left( \boldsymbol{\theta }_{0}\right) \text{IF}\left(
		t_{0}^{l},\boldsymbol{T}_{\beta },\boldsymbol{F}_{\boldsymbol{\theta }%
			_{0}}\right) -\frac{K_{l}}{K}\tsum\limits_{j=1}^{L+1}p_{lj}\left( 
		\boldsymbol{\theta }_{0}\right) ^{\beta -1}\left( -p_{lj}\left( \boldsymbol{%
			\theta }_{0}\right) +\boldsymbol{e}_{t_{0}^{l}}\right) \left( \frac{\partial
			p_{lj}\left( \boldsymbol{\theta }\right) }{\partial \boldsymbol{\theta }}%
		\right) _{\boldsymbol{\theta }=\boldsymbol{\theta }_{0}}=\boldsymbol{0}_{3}
	\end{equation*}%
	and 
	\begin{equation*}
		\text{IF}\left( t_{0}^{l},\boldsymbol{T}_{\beta },\boldsymbol{F}_{\boldsymbol{%
				\theta }_{0}}\right) =\boldsymbol{J}_{\beta }\left( \boldsymbol{\theta }%
		_{0}\right) ^{-1}\frac{K_{l}}{K}\tsum\limits_{j=1}^{L+1}p_{lj}\left( 
		\boldsymbol{\theta }_{0}\right) ^{\beta -1}\left( \frac{\partial
			p_{lj}\left( \boldsymbol{\theta }\right) }{\partial \boldsymbol{\theta }}%
		\right) _{\boldsymbol{\theta }=\boldsymbol{\theta }_{0}}\left( -p_{lj}\left( 
		\boldsymbol{\theta }_{0}\right) +\boldsymbol{e}_{t_{0}^{l}}\right) .
	\end{equation*}
	
	In a similar way can be obtained the IF for all $R$ stress conditions. 
    
    \end{proof}

	\bigskip 
	
	Now we pay special attention to $\text{IF}\left( t_{0}^{l},\boldsymbol{T}_{\beta },%
	\boldsymbol{F}_{\boldsymbol{\theta }_{0}}\right) $. The matrix $\boldsymbol{J%
	}_{\beta }\left( \boldsymbol{\theta }_{0}\right) $ is assumed to be bounded,
	and so the robustness of the estimators depends on the boundedness of the
	second factor of the IF, given by%
	\begin{equation*}
		\frac{K_{l}}{K}\tsum\limits_{j=1}^{L+1}p_{lj}\left( \boldsymbol{\theta }%
		_{0}\right) ^{\beta -1}\left( \frac{\partial p_{lj}\left( \boldsymbol{\theta 
			}\right) }{\partial \boldsymbol{\theta }}\right) _{\boldsymbol{\theta }=%
			\boldsymbol{\theta }_{0}}\left( -p_{lj}\left( \boldsymbol{\theta }%
		_{0}\right) +\boldsymbol{e}_{t_{0}^{l}}\right). 
	\end{equation*}	
	A crucial requirement for robustness of an estimator is that the IF remains bounded. In our formulation, since the data space is
	partitioned into a finite number of inspection intervals $L+1$, partial
	derivatives $\left( \frac{\partial p_{lj}\left( \boldsymbol{\theta }\right) 
	}{\partial \boldsymbol{\theta }}\right) _{\boldsymbol{\theta }=\boldsymbol{%
			\theta }_{0}}$ and the probabilities $p_{lj}\left( \boldsymbol{\theta }%
	_{0}\right) $ are inherently bounded for any valid parameter space.
	Furthermore, the term $p_{lj}\left( \boldsymbol{\theta }_{0}\right) ^{\beta
		-1}$ \ acts as a continuous downweighting mechanism: Although the influence
	function remains unbounded for $0\leq \beta <1$ when some cell probabilities
	become arbitrarily small, the divergence is substantially attenuated for
	every $\beta >0$ compared with the maximum likelihood case ($\beta =0)$.
	Consequently, the influence function exhibits improved robustness properties
	as soon as $\beta $ departs from zero. If $\beta \geq 1$ therefore when cell
	probabilities become arbitrarily small the term $p_{lj}\left( \boldsymbol{%
		\theta }_{0}\right) ^{\beta -1}$ tends to $1$ if $\beta =0$ or $0$ if $\beta
	>1$.

    The influence function analysis confirms that the WMDPDE offers improved robustness over the MLE for every $\beta > 0$, with the degree of robustness increasing with $\beta$. In the following section, we quantify this robustness advantage empirically through a Monte Carlo simulation study.

	\section{Sensitivity Analysis}
	\label{sec:sim_analysis}
	
	We conduct a sensitivity analysis through Monte Carlo simulation study to evaluate the performance of the WMDPDE when we have data contaminated. The study assesses two complementary properties: efficiency under clean data ($\varepsilon = 0$) and robustness relative to the MLE as the contamination proportion $\varepsilon$ increases.
	
	\subsection{Simulation CyALT design}
	\label{subsec:sim_CyALT_design}
	
	The CyALT experiment is designed with two cyclic-stress conditions ($R = 2$). Both conditions share a common floor level $s_F = 0.40$, which exceeds the ceiling level of the normal operating cyclic-stress condition, $s_{0C} = 0.20$, ensuring full acceleration throughout the assumed CyALT design. A total of $K = 200$ units are tested on the elevated ceiling stress levels, $s_{1C} = 0.65$ and $s_{2C} = 1.00$, with cyclic fraction $\tau = 0.40$. These $K = 200$ units are allocated as $K_1 = 120$ and $K_2 = 80$, following proportions
	$\pi_1 = 0.60$ and $\pi_2 = 0.40$. The functional test runs to a censoring time of $t_c = 80$ cycles, with $L = 6$ pre-fixed inspection times at $15, 25, 35, 50, 65$, and $80$ cycles, where the number of failures at these times are recorded. The use condition has an $s_{0F} = 0.00$ and $s_{0C} = 0.20$, and the mission time is $t_0 = 100$ cycles. It is to be noted here that all the stress conditions discussed in our configuration represent standardized stress levels. The configuration under this simulation study is summarized in Figure \ref{fig:sim_design}.
	
	\begin{figure}[htbp]
		\centering
		\resizebox{0.9\textwidth}{!}{%
			\begin{tikzpicture}[>=stealth, scale=1.0]

%
%

\fill[gray!13]  (0,0.0) rectangle (12,1.2);   

\foreach \y in {0.0, 1.2, 2.4, 3.9, 6.0}
  \draw[gray!35, thin, dashed] (0,\y) -- (12,\y);

\draw[gray!55, thick]
  (0,0)--(0,1.2)--(0.8,1.2)--(0.8,0)--(2,0)
  --(2,1.2)--(2.8,1.2)--(2.8,0)--(4,0)
  --(4,1.2)--(4.8,1.2)--(4.8,0)--(6,0)
  --(6,1.2)--(6.8,1.2)--(6.8,0)--(8,0)
  --(8,1.2)--(8.8,1.2)--(8.8,0)--(10,0)
  --(10,1.2)--(10.8,1.2)--(10.8,0)--(12,0);

\draw[blue!65!black, very thick]
  (0,2.4)--(0,3.9)--(0.8,3.9)--(0.8,2.4)--(2,2.4)
  --(2,3.9)--(2.8,3.9)--(2.8,2.4)--(4,2.4)
  --(4,3.9)--(4.8,3.9)--(4.8,2.4)--(6,2.4)
  --(6,3.9)--(6.8,3.9)--(6.8,2.4)--(8,2.4)
  --(8,3.9)--(8.8,3.9)--(8.8,2.4)--(10,2.4)
  --(10,3.9)--(10.8,3.9)--(10.8,2.4)--(12,2.4);

\draw[red!65!black, very thick]
  (0,2.4)--(0,6.0)--(0.8,6.0)--(0.8,2.4)--(2,2.4)
  --(2,6.0)--(2.8,6.0)--(2.8,2.4)--(4,2.4)
  --(4,6.0)--(4.8,6.0)--(4.8,2.4)--(6,2.4)
  --(6,6.0)--(6.8,6.0)--(6.8,2.4)--(8,2.4)
  --(8,6.0)--(8.8,6.0)--(8.8,2.4)--(10,2.4)
  --(10,6.0)--(10.8,6.0)--(10.8,2.4)--(12,2.4);

\node[font=\large, gray!65] at (12.65, 0.6)  {$\cdots$};
\node[font=\large, gray!65] at (12.65, 3.15) {$\cdots$};
\node[font=\large, gray!65] at (12.65, 4.50) {$\cdots$};

\draw[thick, ->] (-0.3,0) -- (13.4,0)
  node[right, font=\small] {Cycle};
\draw[thick, ->] (-0.3,0) -- (-0.3,7.0)
  node[above, font=\small] {Stress $s$};

\foreach \y/\lab in {
  0.0/{$0.00$},
  1.2/{$0.20$},
  2.4/{$0.40$},
  3.9/{$0.65$},
  6.0/{$1.00$}}
{
  \draw (-0.45,\y)--(-0.3,\y);
  \node[left, font=\small] at (-0.48,\y) {\lab};
}

\draw[<->, gray!60, thin] (0,-0.55)--(0.8,-0.55)
  node[midway, below, font=\footnotesize] {$\tau T$};
\draw[<->, gray!60, thin] (0.8,-0.55)--(2,-0.55)
  node[midway, below, font=\footnotesize] {$(1{-}\tau)T$};
\draw[gray!40, thin, dotted] (0.8,-0.80)--(0.8,6.5);  
\draw[decorate,
  decoration={brace, mirror, amplitude=5pt, raise=3pt}]
  (0,-0.85)--(2,-0.85)
  node[midway, below=7pt, font=\footnotesize] {One cycle of $T$ units};

\foreach \xx in {2.25, 3.75, 5.25, 7.50, 9.75, 12.0}
  \draw[gray!45, thin, dashed] (\xx,0)--(\xx,6.5);

\draw[thick, ->] (0,-2.0)--(13.5,-2.0)
  node[right, font=\small] {$t$ (cycles)};
\node[below, font=\footnotesize] at (0,-2.0) {$0$};

\foreach \xx/\lbl/\t in {
  2.25/{$IT_1$}/{15},
  3.75/{$IT_2$}/{25},
  5.25/{$IT_3$}/{35},
  7.50/{$IT_4$}/{50},
  9.75/{$IT_5$}/{65},
  12.0/{$IT_6$}/{80}}
{
  \draw[thick] (\xx,-1.88)--(\xx,-2.12);
  \node[above, font=\footnotesize] at (\xx,-1.88) {\lbl};
  \node[below, font=\footnotesize] at (\xx,-2.12) {\t};
}

\node[right, font=\footnotesize, gray!70]   at (12.0,0.0)  {$s_{0F}$};
\node[right, font=\footnotesize, gray!70]   at (12.0,1.2)  {$s_{0C}$};
\node[right, font=\footnotesize, black!60]  at (12.0,2.4)  {$s_F$};
\node[right, font=\footnotesize, blue!65!black] at (12.0,3.9) {$s_{1C} \rightarrow Cond. 1:\;K_1{=}120$};
\node[right, font=\footnotesize, red!65!black]  at (12.0,6.0) {$s_{2C}\rightarrow Cond. 2:\;K_1{=}80$};


\end{tikzpicture}   
		}
		\caption{Cyclic-stress ALT design for the simulation study. The grey band shows the use-condition stress cycling between $s_{0F} = 0.00$ and $s_{0C} = 0.20$. Vertical dashed lines mark the $L = 6$ inspection times $IT_j \in \{15, 25, 35, 50, 65, 80\}$ cycles, $ j=1,\ldots,6$ with the full test timeline (bottom).}
		\label{fig:sim_design}
	\end{figure}
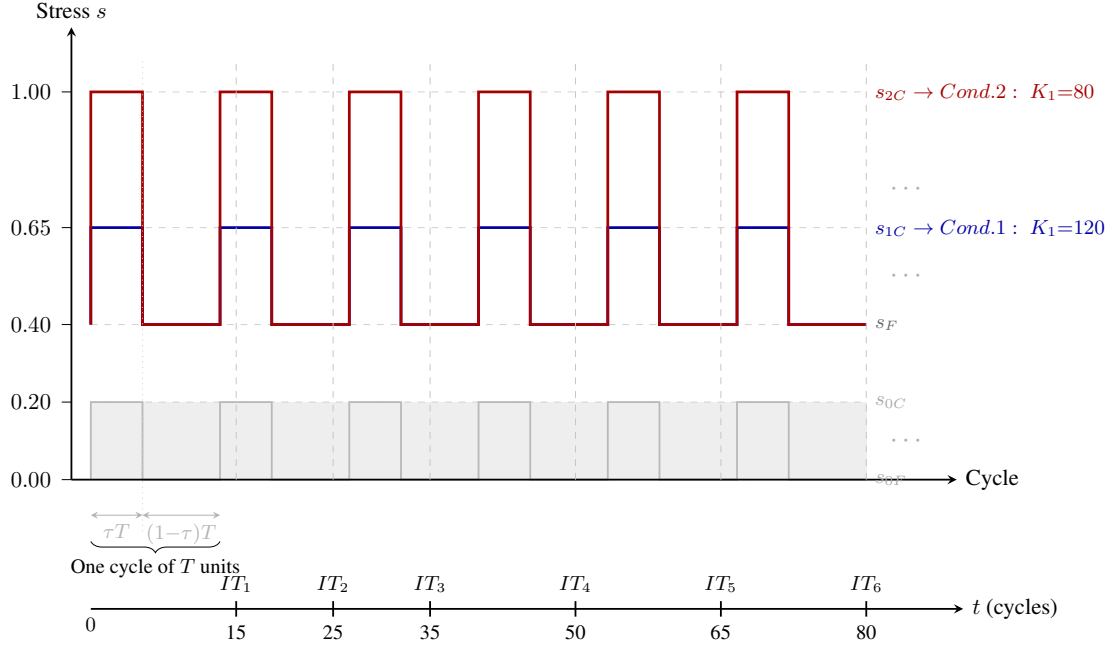
	
	We assume the true parameter vector to be
	$\boldsymbol{\theta}_0 = (\alpha_0, \alpha_1, \sigma)^\top 	= (5.0,\,-2.0,\,0.5)^\top$, which is equivalent to supposing the lifetime of a unit exposed to $s$ level of stress follows a log-normal distribution with a scale parameter  $\zeta (s) = e^{5-2s}$ and shape parameter $\sigma = 0.5$. Under this model, the true lifetime characteristics at the normal cyclic-stress condition are: median lifetime $t_{0.5,0} = 124.016$ cycles, $\mathrm{MTTF}_0 = 140.528$ cycles, and reliability $R_0(100) = 0.667$. 
	
	The WMDPD estimation method is demonstrated for various values of the DPD tuning parameter $\beta \in \{0,\,0.2,\,0.4,\,0.6,\,0.8,\,1.0\}$, where $\beta = 0$ leads to the case of maximum likelihood estimation. The contamination proportion considered ranges over $11$ equally spaced levels such that $\varepsilon \in \{0,\,0.025,\,0.050,\,\ldots,\,0.250\}$. For each $(\varepsilon, \beta)$ combination, $1000$ independent interval-censored datasets are generated under the CyALT design, each consisting of multinomial failure counts for the $R = 2$ stress groups across the $L = 6$ inspection intervals. The WMDPDE is fitted to each dataset to assess the efficiency and robustness trade-off governed by $\beta$ under increasing data contamination. 
	
	
    
    \FloatBarrier
    \subsection{Contamination scheme}
    \label{subsec:contam_scheme}    
    To assess the robustness of the WMDPDE under realistic data contamination, we introduce a proportion $\varepsilon$ of outlying units in Group~1 exhibiting premature failure concentrated in the first three inspection intervals. The contamination follows the two-sample mixture model of \citet{basu1998robust} applied to multinomial interval-censored data. For Group~1, the observed count vector $\boldsymbol{n}_1 = (n_{11},\ldots,n_{1,L+1})$ is the sum of two independent multinomial draws such that $ \boldsymbol{n}_1 \;=\; \boldsymbol{n}_1^{(0)} + \boldsymbol{n}_1^{(c)}$, where $\boldsymbol{n}_1^{(0)} \sim \mathrm{Multinomial} \!\left(\lfloor(1-\varepsilon)K_1\rfloor,\,
    \mathbf{p}_1(\boldsymbol{\theta}_0)\right)$ is drawn from the true lognormal model and $\boldsymbol{n}_1^{(c)} \sim \mathrm{Multinomial}\!\left(K_1 - \lfloor(1-\varepsilon)K_1\rfloor,\, \boldsymbol{q}\right)$ is drawn from the contamination distribution, 
    \begin{equation}
        q_j = \begin{cases} \tfrac{1}{3}, & j = 1, 2, 3, \\[3pt]
                            0,            & j = 4,\ldots,L+1,
              \end{cases}
        \label{eq:q_cont}
    \end{equation}
    which concentrates mass uniformly across the first three inspection intervals $(0,\,15]$, $(15,\,25]$, and $(25,\,35]$.
    Group~2 is always drawn from the true model without contamination. In expectation, the contaminated cell proportions for Group~1 reduce to
    \begin{equation}
        \tilde{p}_{1j}(\varepsilon) =
        (1-\varepsilon)\,p_{1j}(\boldsymbol{\theta}_0) + \varepsilon\,q_j,
        \qquad j = 1,\ldots,L+1,
        \label{eq:cont_prob}
    \end{equation}
    which constitute a valid probability vector since $\sum_{j=1}^{L+1}\tilde{p}_{1j} = 1$.
    
	\begin{table}[htbp]
    \centering
		\small
		\caption{MSE of WMDPDE for the model parameters under multi-cell contamination based on $1000$ Monte Carlo replications with $\boldsymbol{\theta}_0 = (5.0,\,-2.0,\,0.5)^\top$.}
		\begin{tabularx}{\textwidth}{c *{6}{>{\centering\arraybackslash}X}}
			\toprule
			\multirow{2}{*}{$\varepsilon$} & \multirow{2}{*}{MLE} 
			& \multicolumn{5}{c}{$\beta$} \\
			\cmidrule(lr){3-7}
			& & $0.2$ & $0.4$ & $0.6$ & $0.8$ & $1.0$ \\
			\midrule
			\multicolumn{7}{c}{%
				------- $\alpha_0$ -------} \\[3pt]
			$0.000$ & 0.0308 & 0.0314 & 0.0326 & 0.0342 & 0.0361 & 0.0382 \\
			$0.025$ & 0.0379 & 0.0362 & 0.0360 & 0.0366 & 0.0377 & 0.0393 \\
			$0.050$ & 0.0547 & 0.0496 & 0.0469 & 0.0461 & 0.0463 & 0.0472 \\
			$0.075$ & 0.0799 & 0.0693 & 0.0625 & 0.0587 & 0.0571 & 0.0567 \\
			$0.100$ & 0.1265 & 0.1086 & 0.0956 & 0.0872 & 0.0823 & 0.0799 \\
			$0.125$ & 0.1695 & 0.1441 & 0.1243 & 0.1105 & 0.1018 & 0.0967 \\
			$0.150$ & 0.2456 & 0.2130 & 0.1857 & 0.1649 & 0.1505 & 0.1412 \\
			$0.175$ & 0.3196 & 0.2813 & 0.2478 & 0.2208 & 0.2010 & 0.1877 \\
			$0.200$ & 0.4309 & 0.3849 & 0.3426 & 0.3064 & 0.2782 & 0.2580 \\
			$0.225$ & 0.5365 & 0.4852 & 0.4362 & 0.3920 & 0.3554 & 0.3277 \\
			$0.250$ & 0.6838 & 0.6299 & 0.5768 & 0.5275 & 0.4853 & 0.4524 \\
			\midrule
			\multicolumn{7}{c}{%
				------- $ \alpha_1$ -------} \\[3pt]
			$0.000$ & 0.0707 & 0.0723 & 0.0757 & 0.0803 & 0.0860 & 0.0922 \\
			$0.025$ & 0.0849 & 0.0818 & 0.0819 & 0.0844 & 0.0884 & 0.0935 \\
			$0.050$ & 0.1127 & 0.1045 & 0.1013 & 0.1017 & 0.1044 & 0.1085 \\
			$0.075$ & 0.1543 & 0.1362 & 0.1254 & 0.1204 & 0.1195 & 0.1212 \\
			$0.100$ & 0.2324 & 0.2021 & 0.1811 & 0.1687 & 0.1627 & 0.1612 \\
			$0.125$ & 0.3059 & 0.2628 & 0.2303 & 0.2087 & 0.1961 & 0.1899 \\
			$0.150$ & 0.4333 & 0.3774 & 0.3319 & 0.2984 & 0.2764 & 0.2635 \\
			$0.175$ & 0.5539 & 0.4880 & 0.4317 & 0.3877 & 0.3567 & 0.3370 \\
			$0.200$ & 0.7565 & 0.6755 & 0.6028 & 0.5419 & 0.4956 & 0.4637 \\
			$0.225$ & 0.9447 & 0.8520 & 0.7655 & 0.6890 & 0.6268 & 0.5808 \\
			$0.250$ & 1.2110 & 1.1116 & 1.0161 & 0.9291 & 0.8562 & 0.8007 \\
			\midrule
			\multicolumn{7}{c}{%
				------- $\sigma$ -------} \\[3pt]
			$0.000$ & 0.0009 & 0.0009 & 0.0010 & 0.0011 & 0.0012 & 0.0013 \\
			$0.025$ & 0.0011 & 0.0010 & 0.0010 & 0.0010 & 0.0011 & 0.0012 \\
			$0.050$ & 0.0017 & 0.0014 & 0.0013 & 0.0013 & 0.0013 & 0.0013 \\
			$0.075$ & 0.0027 & 0.0023 & 0.0020 & 0.0018 & 0.0018 & 0.0017 \\
			$0.100$ & 0.0039 & 0.0032 & 0.0027 & 0.0024 & 0.0023 & 0.0022 \\
			$0.125$ & 0.0054 & 0.0045 & 0.0037 & 0.0032 & 0.0029 & 0.0027 \\
			$0.150$ & 0.0072 & 0.0063 & 0.0054 & 0.0047 & 0.0042 & 0.0039 \\
			$0.175$ & 0.0083 & 0.0074 & 0.0065 & 0.0057 & 0.0051 & 0.0047 \\
			$0.200$ & 0.0100 & 0.0090 & 0.0080 & 0.0071 & 0.0064 & 0.0059 \\
			$0.225$ & 0.0120 & 0.0111 & 0.0101 & 0.0091 & 0.0082 & 0.0075 \\
			$0.250$ & 0.0130 & 0.0122 & 0.0113 & 0.0104 & 0.0095 & 0.0089 \\
			\bottomrule
		\end{tabularx}
		\label{tab:mse_par}
	\end{table}
    
    This scheme represents a physically realistic scenario in which a proportion of units in Group~1 originate from a substandard production batch exhibiting premature failure, well before the bulk of the population is expected to fail under the assumed lognormal model. For example, at $\varepsilon = 0.250$, the contaminated draw adds $n_{\mathrm{cont}} = K_1 - \lfloor(1-\varepsilon)K_1\rfloor = 120 - \lfloor 0.75 \times 120 \rfloor = 30$ units uniformly across the three early intervals, $10$ units per interval, against a combined expected count of approximately $24$ units under the true model, representing an overcount of roughly $2.2$ times.    
    
	\subsection{MSEs and CIs of the estimators}
    \label{sec_5.3}

    First, the MSE of the WMDPDE is reported for the model parameters $\alpha_0$, $\alpha_1$, $\sigma$ and for the lifetime characteristics $t_{0.5,0}$, $\mathrm{MTTF}_0$, and $R_0(t_0)$ while varying the contamination proportion $\varepsilon$ and DPD tuning parameter $\beta$, summarizing the point estimation accuracy. Second, three types of $95\%$ confidence intervals from Section \ref{subsec:CIs} are constructed, and their coverage probability (CP) and average width (AW) are evaluated: direct asymptotic CIs, transformed asymptotic CIs (log-transformation for $t_{0.5,0}$ and $\mathrm{MTTF}_0$; logit-transformation for $R_0(t_0)$), and BCa bootstrap CIs. Together, these results provide a comprehensive picture of the efficiency and robustness trade-off governed by tuning parameter $\beta$ under increasing data contamination.
	
	\begin{table}[htbp]
		\centering
		\small
		\caption{RMSE of WMDPDE for lifetime characteristics under multi-cell contamination based on $1000$ Monte Carlo replications with $\boldsymbol{\theta}_0 = (5.0,\,-2.0,\,0.5)^\top$.}
		\begin{tabularx}{\textwidth}{c *{6}{>{\centering\arraybackslash}X}}
			\toprule
			\multirow{2}{*}{$\varepsilon$} & \multirow{2}{*}{MLE} 
			& \multicolumn{5}{c}{$\beta$} \\
			\cmidrule(lr){3-7}
			& & $0.2$ & $0.4$ & $0.6$ & $0.8$ & $1.0$ \\
			\midrule
			\multicolumn{7}{c}{%
				------- Median $ (t_{0.5,0})$ -------} \\[3pt]
			$0.000$ & 347.4381 & 356.6034 & 372.6519 & 393.9554 & 419.7740 & 449.0195 \\
			$0.025$ & 357.0347 & 350.2933 & 355.0651 & 367.2133 & 384.2778 & 404.9652 \\
			$0.050$ & 473.3791 & 437.1889 & 421.4191 & 420.6192 & 429.3586 & 443.7770 \\
			$0.075$ & 671.3932 & 594.2522 & 545.5111 & 521.2210 & 514.2611 & 518.2979 \\
			$0.100$ & 995.5842 & 869.7040 & 776.2879 & 715.9490 & 682.0797 & 666.5478 \\
			$0.125$ & 1295.4714 & 1122.7661 & 981.7642 & 880.3063 & 815.2676 & 777.0813 \\
			$0.150$ & 1768.3120 & 1567.1260 & 1388.4503 & 1245.3243 & 1142.1188 & 1073.1247 \\
			$0.175$ & 2197.9553 & 1982.7810 & 1782.5693 & 1612.7421 & 1483.0609 & 1392.1557 \\
			$0.200$ & 2780.0249 & 2545.6960 & 2316.5644 & 2109.4309 & 1940.2288 & 1813.9268 \\
			$0.225$ & 3267.2996 & 3032.0126 & 2792.8151 & 2564.3858 & 2364.8444 & 2206.3321 \\
			$0.250$ & 3921.1748 & 3699.5514 & 3466.8235 & 3235.6640 & 3025.1280 & 2851.9708 \\
			\midrule
			\multicolumn{7}{c}{%
				------- $\mathrm{MTTF}_0$ -------} \\[3pt]
			$0.000$ & 458.8328 & 470.1590 & 489.6857 & 515.3863 & 546.0828 & 580.3937 \\
			$0.025$ & 459.0539 & 453.3814 & 460.6779 & 476.5136 & 498.0512 & 523.6371 \\
			$0.050$ & 580.2460 & 547.2040 & 535.6243 & 539.5123 & 553.1868 & 572.6097 \\
			$0.075$ & 773.3964 & 698.4717 & 653.7853 & 634.2309 & 632.2501 & 641.4372 \\
			$0.100$ & 1124.7854 & 996.8751 & 904.4849 & 846.5311 & 815.6251 & 803.3544 \\
			$0.125$ & 1430.1391 & 1251.7100 & 1109.4504 & 1009.0429 & 945.9084 & 909.9079 \\
			$0.150$ & 1942.6814 & 1727.8183 & 1540.9518 & 1393.5490 & 1288.2982 & 1218.4631 \\
			$0.175$ & 2439.3180 & 2203.7951 & 1988.9509 & 1809.4981 & 1673.8189 & 1579.2605 \\
			$0.200$ & 3103.6104 & 2840.5563 & 2588.1441 & 2363.2467 & 2181.1127 & 2045.5002 \\
			$0.225$ & 3660.9440 & 3391.0746 & 3121.6033 & 2868.1864 & 2649.2503 & 2476.3103 \\
			$0.250$ & 4455.8094 & 4195.8240 & 3927.2461 & 3664.3877 & 3427.4587 & 3233.3524 \\
			\midrule
			\multicolumn{7}{c}{%
				------- $R_0(t_0)$ -------} \\[3pt]
			$0.000$ & 0.0115 & 0.0117 & 0.0121 & 0.0126 & 0.0132 & 0.0139 \\
			$0.025$ & 0.0154 & 0.0145 & 0.0142 & 0.0143 & 0.0147 & 0.0152 \\
			$0.050$ & 0.0235 & 0.0212 & 0.0198 & 0.0192 & 0.0191 & 0.0193 \\
			$0.075$ & 0.0343 & 0.0299 & 0.0269 & 0.0251 & 0.0242 & 0.0238 \\
			$0.100$ & 0.0528 & 0.0460 & 0.0408 & 0.0372 & 0.0351 & 0.0339 \\
			$0.125$ & 0.0687 & 0.0595 & 0.0520 & 0.0464 & 0.0428 & 0.0405 \\
			$0.150$ & 0.0940 & 0.0835 & 0.0741 & 0.0666 & 0.0611 & 0.0574 \\
			$0.175$ & 0.1173 & 0.1061 & 0.0956 & 0.0866 & 0.0798 & 0.0749 \\
			$0.200$ & 0.1475 & 0.1356 & 0.1238 & 0.1131 & 0.1043 & 0.0977 \\
			$0.225$ & 0.1722 & 0.1605 & 0.1485 & 0.1369 & 0.1268 & 0.1186 \\
			$0.250$ & 0.2063 & 0.1955 & 0.1840 & 0.1724 & 0.1618 & 0.1529 \\
			\bottomrule
		\end{tabularx}
		\label{tab:mse_char}
	\end{table}
    
    Tables \ref{tab:mse_par} and \ref{tab:mse_char} report the MSE/RMSE of the WMDPDE for the model parameters and lifetime characteristics. Under clean data, the MLE achieves the lowest MSE for all quantities, as expected from its asymptotic efficiency. The 
    MSE increases monotonically with $\beta$, reflecting the small efficiency cost of robustification, which is expected and consistent with the general properties of the DPD family.

    As contamination increases, the picture of this ordering changes. The MLE becomes increasingly sensitive to contamination, with its MSE rising sharply, while the WMDPDE with larger $\beta$ remains more stable. This shift occurs even at low contamination levels and holds consistently across both tables and confirms the robustness advantage of the WMDPDE under contaminated data.
    \FloatBarrier
    \begin{figure}[ht]
		\centering
		{\small Model Parameters}\\[0.75em]
		\includegraphics[width=1\linewidth]{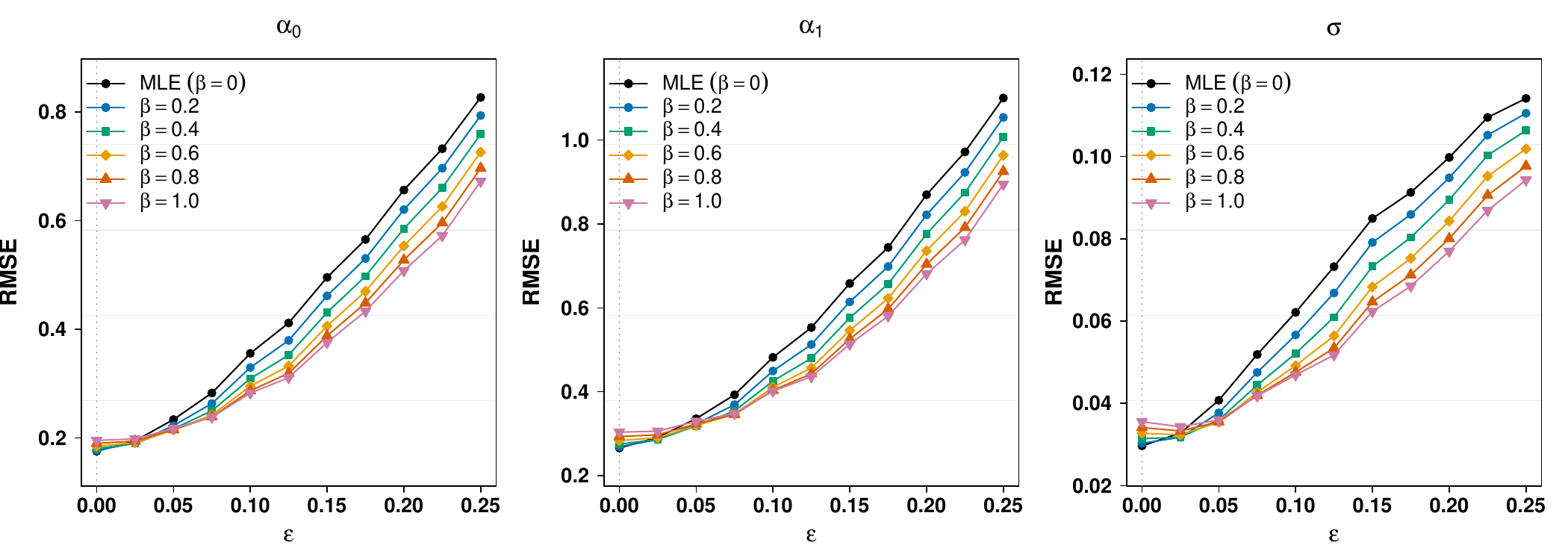}\\[1em]
		{\small Lifetime Characteristics}\\ [0.75em]
		\includegraphics[width=1\linewidth]{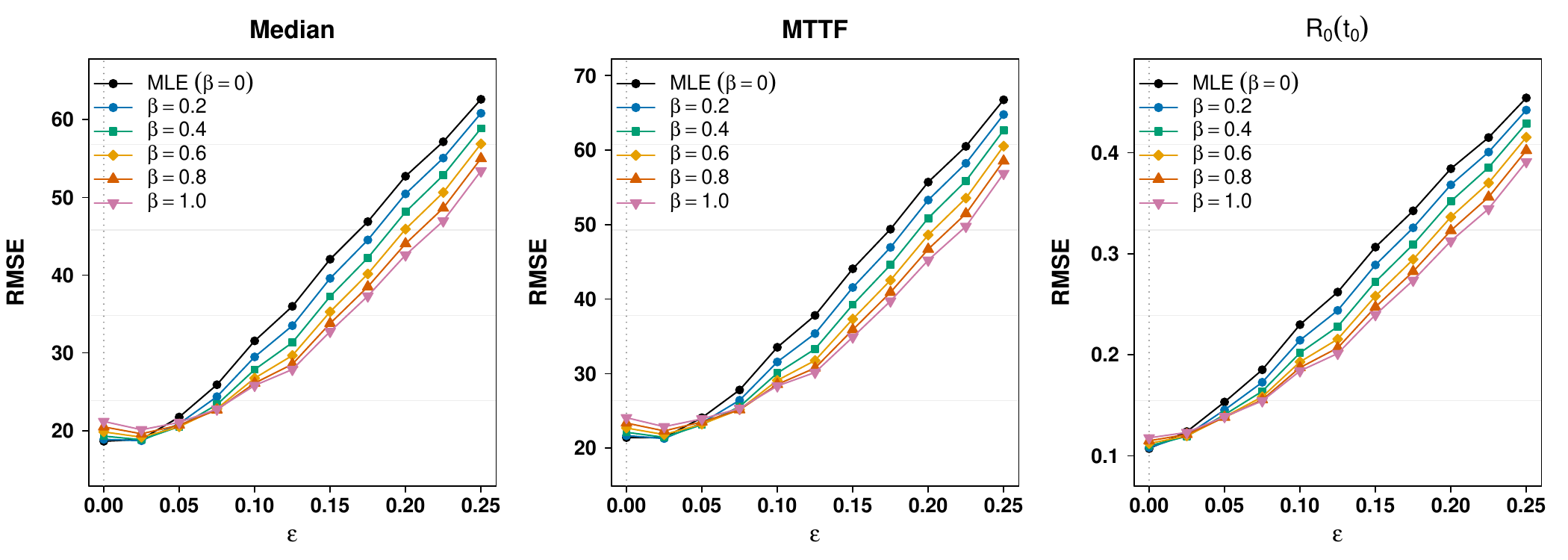}
		\caption{RMSE of the MLE ($\beta=0$) and
			WMDPDE ($\beta \in \{0.2, 0.4, 0.6, 0.8, 1.0\}$) with respect to the contamination proportion $\varepsilon$ for the model parameters
			$\alpha_0$, $\alpha_1$, $\sigma$ (top row) and the lifetime characteristics, median $t_{0.5,0}$, mean time to failure
			$\mathrm{MTTF}_0$, and reliability $R_0(t_0)$, at normal operating cyclic-stress condition (bottom row).}
		\label{fig:rmse_par_char}
	\end{figure}	    
    Figure \ref{fig:rmse_par_char} illustrates the efficiency and robustness trade-off across the wider range of $\varepsilon$. The crossover structure is clearly visible for every quantity shown in the panels. The MLE starts with the lowest RMSE at $\varepsilon = 0$ but increases rapidly as contamination grows, while estimators with larger $\beta$ rise more slowly. The results suggest that even a small positive $\beta$ provides meaningful robustness with negligible loss of efficiency in practice.

    The simulation study confirms the practical advantages of the WMDPDE under data contamination. We now illustrate the full methodology on a realistic engineering application.
	
	\section{Implementation on the air-conditioner data}
	\label{sec:real_analysis}
    
	We illustrate the proposed methodology using the automotive air-conditioner evaporator application introduced in Section \ref{sec:industry_app}. We use preliminary data obtained under an accelerated cyclic-stress loading, originally presented by \cite{kim2021optimal}, to simulate a dataset under the proposed CyALT experiment. The analysis follows three steps: (1) preliminary parameter estimation from the initial experiment, (2) design of a CyALT experiment, and (3) application of the WMDPDE to the resulting simulated interval-monitored data.
	
	\subsection{Preliminary parameter estimation}
	\label{sec:prelim_est}
	
	The evaporator units operate under a cyclic refrigerant pressure load in which the compressor alternates between a floor pressure and a ceiling pressure, as shown in Figure \ref{fig:ac}. Each cycle lasts approximately $20$ seconds, with the ceiling pressure sustained for half the cycle duration, giving $\tau = 0.5$. The use condition has floor pressure $v_{0F} = 0.1$~MPa and ceiling pressure $v_{0C} = 0.25$~MPa. Since the stress variable is refrigerant pressure, the inverse power law model is the natural stress-life relationship, corresponding to $g(v) = \ln(v)$ in equation \eqref{eq:stress_life_original}. Applying the standardization from equation \eqref{eq:standardisation} gives
	\begin{equation}
		s_{im} = \frac{\ln(v_{im}) - \ln(v_{0F})}
		{\ln(v_{hC}) - \ln(v_{0F})} \in [0,1],
		\label{eq:standardise_app}
	\end{equation}
	yielding the standardized use stress levels $s_{0F} = 0.00$ and $s_{0C} = 0.27$.
	
	A preliminary reliability experiment was conducted on $N = 20$ evaporator units under a single accelerated cyclic-stress condition
	with floor level $s_{hF} = 0.50$ and ceiling level $s_{hC} = 1.00$, corresponding to the floor and maximum test pressure as  $v_{hF} = 0.47$~MPa and $v_{hC} = 3.0$~MPa, respectively, with censoring time $t_c = 50{,}000$ cycles. Out of the $N = 20$ units placed
	on test, $18$ failed before $t_c$ where the two surviving units were right-censored at $50{,}000$ cycles. The observed failure times
	(in cycles) are listed in Table \ref{table1}.
	
	\begin{table}[htbp]
		\centering
		\caption{Observed failure times (in cycles) from the preliminary experiment ($N = 20$, $s_{hF} = 0.50$, 
			$s_{hC} = 1.00$, $t_c = 50{,}000$ cycles).}
		\label{table1}
		\small
        \begin{tabularx}{\textwidth}{c *{8}{>{\centering\arraybackslash}X}}
			\toprule
			\multicolumn{9}{c}{Failure time} \\
			\midrule
			26763  & 31959  & 32887  & 	33069  & 34019  & 34924 & 36754  & 37054  & 37385  \\
			38045  & 41033  & 41755  &  42333  & 42818  & 44638  &  44867  & 48364  & 49767  \\
			\bottomrule
		\end{tabularx}
	\end{table}
	
	We fit a lognormal distribution to the $18$ observed failure times given in Table \ref{tab:observed_data} using the maximum likelihood estimation (\texttt{fitdistr} in R) that yields $\sigma = 0.156$. We adopt it as the pre-estimate of the shape parameter for the main experiment. The pre-estimates of the remaining parameters $\alpha_0$ and $\alpha_1$ are also uniquely obtainable given the proportion of failures at any two cyclic-stress conditions. Since we have two further inputs: the observed failure proportion $p_h = 0.90 (= 18/20 )$ under the preliminary cyclic-stress condition within $t_c$, and the experimenter's guess according to \citep{kim2021optimal} for the probability of failure at the use condition $p_u = 0.001$ within the same period.
	
	The location parameters at the preliminary and use conditions are recovered by inverting the lognormal CDF in equation \eqref{p_i_l+1}
	\begin{align}
		\mu_h &= \ln t_c - \sigma\,\Phi^{-1}(p_h) = 10.620,
		\label{eq:mu_h}\\
		\mu_u &= \ln t_c - \sigma\,\Phi^{-1}(p_u) = 11.302.
		\label{eq:mu_u}
	\end{align}
	Equating these to the expressions $\mu_m = -\ln\!\left[\tau\,e^{-\alpha_0 - \alpha_1 s_{mC}} + (1-\tau)\,e^{-\alpha_0 - \alpha_1 s_{mF}}\right]$ for $m \in \{h, u\}$ and solving for $\alpha_0$ and $\alpha_1$ yields the remaining pre-estimates as  $11.457$ and $ -1.068$, respectively. Hence, the resulting parameter vector $\boldsymbol{\theta}_0 = (11.457,\,-1.068,\,0.156)^\top$ serves as the true parameter vector for the subsequent analysis.
	
	\subsection{Simulated data under CyALT experiment}
	\label{subsec:real_cyalt_design}
	
	Using the parameter vector $\boldsymbol{\theta}_0$ estimated in Section \ref{sec:prelim_est}, we design a CyALT experiment with $R = 2$
	cyclic-stress conditions. Both conditions share a common floor level $s_F = 0.30$, which is just above the use-condition ceiling
	$s_{0C} = 0.27$, ensuring that even the minimum test stress level exceeds normal operating levels. The ceiling levels are set to
	$s_{1C} = 0.70$ and $s_{2C} = 1.00$, with cyclic fraction $\tau = 0.50$.
	
	A total of $K = 200$ units are allocated across the two groups as $K_1 = 140$ and $K_2 = 60$, following the optimal proportions
	$\pi_1 = 0.70$ and $\pi_2 = 0.30$. The test runs to a censoring time of $t_c = 65{,}000$ cycles, with $L = 6$
	pre-fixed inspection times at $IT_1 = 25{,}000$, $IT_2 = 35{,}000$, $IT_3 = 45{,}000$, $IT_4 = 50{,}000$, $IT_5 = 60{,}000$, and
	$IT_6 = t_c = 65{,}000$ cycles.     
	\begin{figure}[h]
		\centering
		\resizebox{0.9\textwidth}{!}{%

\begin{tikzpicture}[>=stealth, scale=1.0]
	
	%
	%
	
	\fill[gray!13] (0,0.0) rectangle (12,1.62);
	
	\foreach \y in {0.0, 1.62, 1.80, 4.20, 6.0}
	\draw[gray!35, thin, dashed] (0,\y) -- (12,\y);
	
	\draw[gray!55, thick]
	(0,0)--(0,1.62)--(1,1.62)--(1,0)--(2,0)
	--(2,1.62)--(3,1.62)--(3,0)--(4,0)
	--(4,1.62)--(5,1.62)--(5,0)--(6,0)
	--(6,1.62)--(7,1.62)--(7,0)--(8,0)
	--(8,1.62)--(9,1.62)--(9,0)--(10,0)
	--(10,1.62)--(11,1.62)--(11,0)--(12,0);
	
	\draw[blue!65!black, very thick]
	(0,1.80)--(0,4.20)--(1,4.20)--(1,1.80)--(2,1.80)
	--(2,4.20)--(3,4.20)--(3,1.80)--(4,1.80)
	--(4,4.20)--(5,4.20)--(5,1.80)--(6,1.80)
	--(6,4.20)--(7,4.20)--(7,1.80)--(8,1.80)
	--(8,4.20)--(9,4.20)--(9,1.80)--(10,1.80)
	--(10,4.20)--(11,4.20)--(11,1.80)--(12,1.80);
	
	\draw[red!65!black, very thick]
	(0,1.80)--(0,6.0)--(1,6.0)--(1,1.80)--(2,1.80)
	--(2,6.0)--(3,6.0)--(3,1.80)--(4,1.80)
	--(4,6.0)--(5,6.0)--(5,1.80)--(6,1.80)
	--(6,6.0)--(7,6.0)--(7,1.80)--(8,1.80)
	--(8,6.0)--(9,6.0)--(9,1.80)--(10,1.80)
	--(10,6.0)--(11,6.0)--(11,1.80)--(12,1.80);
	
	\node[font=\large, gray!65] at (12.65, 0.81) {$\cdots$};
	\node[font=\large, gray!65] at (12.65, 3.00) {$\cdots$};
	\node[font=\large, gray!65] at (12.65, 4.50) {$\cdots$};
	
	\draw[thick, ->] (-0.3,0) -- (13.4,0)
	node[right, font=\small] {Cycle};
	\draw[thick, ->] (-0.3,0) -- (-0.3,7.0)
	node[above, font=\small] {Stress $s$};
	
	\draw (-0.45, 0.0)  -- (-0.3, 0.0);
	\draw (-0.45, 1.62) -- (-0.3, 1.62);
	\draw (-0.45, 1.80) -- (-0.3, 1.80);
	\draw (-0.45, 4.20) -- (-0.3, 4.20);
	\draw (-0.45, 6.0)  -- (-0.3, 6.0);
	
	\node[left, font=\small]               at (-0.48, 0.0)  {$0.00$};
	\node[left, font=\scriptsize, yshift= 4pt] at (-0.48, 1.62) {$0.27$};
	\node[left, font=\scriptsize, yshift=-4pt] at (-0.48, 1.80) {$0.30$};
	\node[left, font=\small]               at (-0.48, 4.20) {$0.70$};
	\node[left, font=\small]               at (-0.48, 6.0)  {$1.00$};
	
	\draw[<->, gray!60, thin] (0,-0.55)--(1,-0.55)
	node[midway, below, font=\footnotesize] {$\tau T$};
	\draw[<->, gray!60, thin] (1,-0.55)--(2,-0.55)
	node[midway, below, font=\footnotesize] {$(1{-}\tau)T$};
	\draw[gray!40, thin, dotted] (1,-0.80)--(1,6.5);
	\draw[decorate,
	decoration={brace, mirror, amplitude=5pt, raise=3pt}]
	(0,-0.85)--(2,-0.85)
	node[midway, below=7pt, font=\footnotesize] {One cycle of $20$ seconds};
	
	\foreach \xx in {4.615, 6.462, 8.308, 9.231, 11.077, 12.0}
	\draw[gray!45, thin, dashed] (\xx,0)--(\xx,6.5);
	
	\draw[thick, ->] (0,-2.0)--(13.5,-2.0)
	node[right, font=\small] {$t$ (${\times}10^3$ cycles)};
	\node[below, font=\footnotesize] at (0,-2.0) {$0$};
	
	\foreach \xx/\lbl/\t in {
		4.615/{$IT_1$}/{25},
		6.462/{$IT_2$}/{35},
		8.308/{$IT_3$}/{45},
		9.231/{$IT_4$}/{50},
		11.077/{$IT_5$}/{60},
		12.0/{$IT_6$}/{65}}
	{
		\draw[thick] (\xx,-1.88)--(\xx,-2.12);
		\node[above, font=\footnotesize] at (\xx,-1.88) {\lbl};
		\node[below, font=\footnotesize] at (\xx,-2.12) {\t};
	}
	
	\node[right, font=\footnotesize, gray!70]       at (12.0, 0.0)
	{$s_{0F}$};
	\node[right, font=\scriptsize,   gray!70, yshift= 4pt] at (12.0, 1.62)
	{$s_{0C}$};
	\node[right, font=\scriptsize,   black!60, yshift=-4pt] at (12.0, 1.80)
	{$s_F$};
	\node[right, font=\footnotesize, blue!65!black] at (12.0, 4.20)
	{$s_{1C} \rightarrow \text{Cond.~1:}\;K_1{=}140$};
	\node[right, font=\footnotesize, red!65!black]  at (12.0, 6.0)
	{$s_{2C} \rightarrow \text{Cond.~2:}\;K_2{=}60$};
	
\end{tikzpicture}   
		}
		\caption{Cyclic-stress ALT design for the real data analysis. The grey band shows the use-condition stress cycling between
			$s_{0F} = 0.00$ and $s_{0C} = 0.27$. Vertical dashed lines mark the $L = 6$ inspection times $IT_j \in \{25, 35, 45, 50, 60, 65\}$ in $1000$ cycles, $ j=1,\ldots,6$ with the full test timeline (bottom).}
		\label{fig:real_design}
	\end{figure}
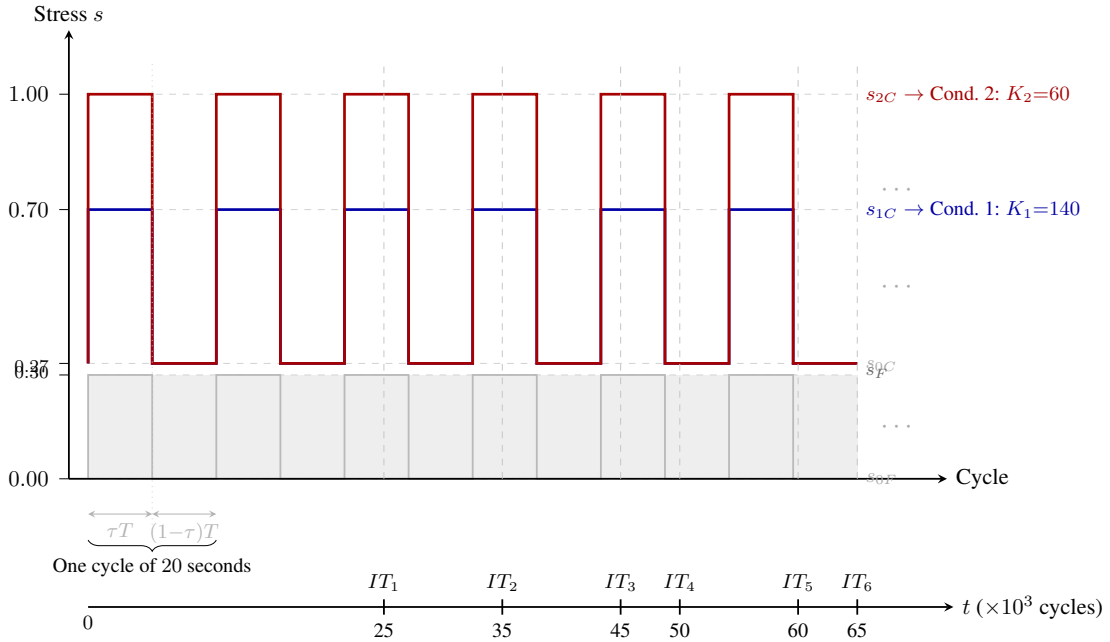

    The stress profiles for both cyclic-stress conditions and the inspection time schedule are illustrated in Figure \ref{fig:real_design}.
    
	\begin{table}[htbp]
		\centering
        \small
		\caption{Interval-censored failure counts from the simulated CyALT experiment with $K_1 = 140$, $K_2 = 60$, and $t_c = 65{,}000$ cycles.}
		\label{tab:observed_data}
		\begin{tabular}{ccc}
			\toprule
			Inspection interval 
			& Group 1 
			& Group 2 \\
			(cycles)
			& ($s_{1C} = 0.70$)
			& ($s_{2C} = 1.00$) \\
			\midrule
			$(0,\;25{,}000]$          &  0  &  0  \\
			$(25{,}000,\;35{,}000]$   &  1  &  3  \\
			$(35{,}000,\;45{,}000]$   & 13  & 30  \\
			$(45{,}000,\;50{,}000]$   & 25  & 12  \\
			$(50{,}000,\;60{,}000]$   & 62  & 14  \\
			$(60{,}000,\;65{,}000]$   & 21  &  1  \\
			Survivors ($>65{,}000$)   & 18  &  0  \\
			\midrule
			Total                      & 140 & 60  \\
			\bottomrule
		\end{tabular}
	\end{table}
	
	Exact lognormal failure times are generated from $\boldsymbol{\theta}_0$ setting the seed as $125$ and binned into the inspection intervals to produce the interval-censored counts shown in Table \ref{tab:observed_data}. Group~1, under ceiling stress $s_{1C} = 0.70$, concentrates most failures in the interval $(50{,}000,\,60{,}000]$ with $18$ survivors beyond $t_c$. Group~2, under the higher ceiling stress $s_{2C} = 1.00$, fails almost entirely before $45{,}000$ cycles with no survivors. This pattern is consistent with the acceleration effect at the higher ceiling stress.
	
	\subsection{WMDPDE estimates and confidence intervals}
	\label{subsec:sim_WMDPDE_estimation}
	
	We are going to fit now the WMDPDE for $\beta \in \{0,\,0.2,\,0.4,\,0.6,\,0.8,\,1.0\}$ by minimizing $H_\beta(\boldsymbol{\theta})$ in equation \eqref{eq:DPD_full_reduced} via Nelder-Mead optimization using the \texttt{optimx} package in R. Parameter estimates with $95\%$ direct asymptotic confidence intervals and $95\%$ BCa bootstrap confidence intervals based on $500$ parametric replicates are reported in Table \ref{tab:param_estimates}.
    
    \noindent The estimates are stable across the full range of $\beta$. The intercept $\widehat{\alpha}_{0,\beta}$ varies from $11.466$ to $11.483$, and the slope $\widehat{\alpha}_{1,\beta}$ from $-1.069$ to $-1.095$, a range of less than $0.03$ units. All confidence intervals contain the true parameter values. The shape parameter $\widehat{\sigma}_\beta$ remains essentially stable, ranging from
	$0.155$ to $0.154$ across $\beta$, close to the pre-estimated value $\sigma = 0.156$. The mild widening of both asymptotic and bootstrap confidence intervals with $\beta$ reflects the small efficiency cost of robustification under clean data, consistent with the general properties of the WMDPDE family.
	
	Table \ref{tab:lifetime_chars} reports three lifetime characteristics at the use condition ($s_{0F} = 0.00$, $s_{0C} = 0.27$): the median lifetime $t_{0.5,0}$, the mean time to failure $\mathrm{MTTF}_0 = \exp(\mu_0 + \sigma^2/2)$, and the
	reliability at mission time $t_0 = 75{,}000$ cycles. For each $\beta$ we report the point estimate, the $95\%$ direct and transformed (log-transformed (median, MTTF) or logit-transformed (reliability)) asymptotic confidence interval, and the $95\%$ bias-corrected and accelerated (BCa) bootstrap confidence interval based on $500$ parametric replicates.
	
	\begin{table}[htbp]
		\centering
        \small
		\caption{WMDPDE estimates of model parameters and $95\%$ confidence intervals from the simulated CyALT evaporator dataset.}
		\label{tab:param_estimates}
		\setlength{\tabcolsep}{5pt}
        \begin{tabularx}{\textwidth}{c *{4}{>{\centering\arraybackslash}X}}
			\toprule
			Parameter & $\beta$ & Estimate
			& Direct 
			& BCa  \\
			\midrule
			$\alpha_0$ & \multicolumn{4}{l}{%
				\hfill True value, $\alpha_{0,0}= 11.4565$} \\[3pt]
			& $0$   & $11.4663$ & $(11.3337,\;11.5990)$
			& $(11.3209,\;11.5710)$ \\
			& $0.2$ & $11.4699$ & $(11.3363,\;11.6034)$
			& $(11.3364,\;11.6084)$ \\
			& $0.4$ & $11.4728$ & $(11.3355,\;11.6101)$
			& $(11.3212,\;11.5927)$ \\
			& $0.6$ & $11.4758$ & $(11.3336,\;11.6180)$
			& $(11.3171,\;11.6150)$ \\
			& $0.8$ & $11.4791$ & $(11.3318,\;11.6264)$
			& $(11.2996,\;11.6007)$ \\
			& $1.0$ & $11.4826$ & $(11.3301,\;11.6351)$
			& $(11.3031,\;11.6249)$ \\
			\midrule
			$\alpha_1$ & \multicolumn{4}{l}{%
				\hfill True value, $\alpha_{1,0}= -1.0684$} \\[3pt]
			& $0$   & $-1.0690$ & $(-1.2831,\;-0.8549)$
			& $(-1.2701,\;-0.8533)$ \\
			& $0.2$ & $-1.0727$ & $(-1.2889,\;-0.8565)$
			& $(-1.2818,\;-0.8544)$ \\
			& $0.4$ & $-1.0769$ & $(-1.3006,\;-0.8532)$
			& $(-1.2767,\;-0.8709)$ \\
			& $0.6$ & $-1.0821$ & $(-1.3154,\;-0.8488)$
			& $(-1.3215,\;-0.8204)$ \\
			& $0.8$ & $-1.0882$ & $(-1.3315,\;-0.8449)$
			& $(-1.3031,\;-0.8436)$ \\
			& $1.0$ & $-1.0948$ & $(-1.3481,\;-0.8415)$
			& $(-1.3287,\;-0.8051)$ \\
			\midrule
			$\sigma$ & \multicolumn{4}{l}{%
				\hfill True value, $\sigma_{0}= 0.1560$} \\[3pt]
			& $0$   & $0.1547$ & $(0.1367,\;0.1727)$
			& $(0.1363,\;0.1736)$ \\
			& $0.2$ & $0.1537$ & $(0.1353,\;0.1721)$
			& $(0.1353,\;0.1749)$ \\
			& $0.4$ & $0.1535$ & $(0.1344,\;0.1727)$
			& $(0.1365,\;0.1757)$ \\
			& $0.6$ & $0.1538$ & $(0.1338,\;0.1737)$
			& $(0.1364,\;0.1771)$ \\
			& $0.8$ & $0.1541$ & $(0.1333,\;0.1748)$
			& $(0.1359,\;0.1747)$ \\
			& $1.0$ & $0.1544$ & $(0.1327,\;0.1760)$
			& $(0.1316,\;0.1765)$ \\
			\bottomrule
		\end{tabularx}
	\end{table}
	\begin{table}[htbp]
		\centering
        \small
		\caption{WMDPDE estimates of lifetime characteristics at the use condition with $95\%$ direct, transformed and BCa bootstrap confidence intervals.}
		\label{tab:lifetime_chars}
		\setlength{\tabcolsep}{4pt}
        \begin{tabularx}{\textwidth}{c *{4}{>{\centering\arraybackslash}X}}
			\toprule
			$\beta$ & Estimate
			& Direct 95\%~CI
			& Trans.\ 95\%~CI
			& BCa 95\%~CI \\
			\midrule
			& \multicolumn{4}{l}{%
				Median lifetime, $t_{0.5,0}$ (cycles) \hfill True value $= 80{,}971$} \\[3pt]
			$0$   & $81{,}766$ & $(73{,}570,\;89{,}963)$
			& $(73{,}968,\;90{,}387)$
			& $(73{,}373,\;90{,}446)$ \\
			$0.2$ & $82{,}011$ & $(73{,}742,\;90{,}279)$
			& $(74{,}145,\;90{,}710)$
			& $(74{,}520,\;89{,}213)$ \\
			$0.4$ & $82{,}196$ & $(73{,}695,\;90{,}697)$
			& $(74{,}120,\;91{,}152)$
			& $(73{,}787,\;90{,}634)$ \\
			$0.6$ & $82{,}378$ & $(73{,}581,\;91{,}176)$
			& $(74{,}034,\;91{,}663)$
			& $(72{,}770,\;90{,}210)$ \\
			$0.8$ & $82{,}573$ & $(73{,}459,\;91{,}687)$
			& $(73{,}944,\;92{,}209)$
			& $(74{,}261,\;91{,}940)$ \\
			$1.0$ & $82{,}779$ & $(73{,}344,\;92{,}214)$
			& $(73{,}862,\;92{,}772)$
			& $(74{,}051,\;92{,}544)$ \\
			\midrule
			&\multicolumn{4}{l}{%
				Mean time to failure, $\mathrm{MTTF}_0$ (cycles) \hfill True value $= 81{,}962$} \\[3pt]
			$0$   & $82{,}751$ & $(74{,}441,\;91{,}060)$
			& $(74{,}845,\;91{,}492)$
			& $(74{,}843,\;90{,}339)$ \\
			$0.2$ & $82{,}985$ & $(74{,}608,\;91{,}363)$
			& $(75{,}017,\;91{,}800)$
			& $(74{,}958,\;91{,}764)$ \\
			$0.4$ & $83{,}171$ & $(74{,}564,\;91{,}778)$
			& $(74{,}994,\;92{,}239)$
			& $(75{,}065,\;92{,}230)$ \\
			$0.6$ & $83{,}358$ & $(74{,}453,\;92{,}263)$
			& $(74{,}912,\;92{,}756)$
			& $(74{,}423,\;92{,}666)$ \\
			$0.8$ & $83{,}558$ & $(74{,}333,\;92{,}784)$
			& $(74{,}824,\;93{,}313)$
			& $(72{,}983,\;92{,}463)$ \\
			$1.0$ & $83{,}771$ & $(74{,}217,\;93{,}325)$
			& $(74{,}741,\;93{,}891)$
			& $(74{,}961,\;93{,}368)$ \\
			\midrule
			& \multicolumn{4}{c}{%
				Reliability $R_0(t_0)$ at $t_0 = 75{,}000$ cycles \hfill True value $= 0.6883$} \\[3pt]
			$0$   & $0.7117$ & $(0.4905,\;0.9329)$
			& $(0.4566,\;0.8788)$
			& $(0.4177,\;0.8724)$ \\
			$0.2$ & $0.7195$ & $(0.4980,\;0.9410)$
			& $(0.4613,\;0.8849)$
			& $(0.4760,\;0.8998)$ \\
			$0.4$ & $0.7246$ & $(0.4986,\;0.9507)$
			& $(0.4588,\;0.8910)$
			& $(0.4872,\;0.9054)$ \\
			$0.6$ & $0.7292$ & $(0.4976,\;0.9608)$
			& $(0.4545,\;0.8969)$
			& $(0.4316,\;0.9089)$ \\
			$0.8$ & $0.7338$ & $(0.4969,\;0.9707)$
			& $(0.4505,\;0.9026)$
			& $(0.4349,\;0.9055)$ \\
			$1.0$ & $0.7387$ & $(0.4970,\;0.9804)$
			& $(0.4470,\;0.9081)$
			& $(0.4542,\;0.9210)$ \\
			\bottomrule
		\end{tabularx}
	\end{table}
	
	Several observations are noteworthy when estimating the characteristics. First, all point estimates are close to the true values. The estimates lie slightly above the true values in all cases, reflecting the small upward bias of the WMDPDE under
	clean data as $\beta$ increases. Second, all confidence intervals, both asymptotic and BCa bootstrap, contain the true value for every $\beta$ and every characteristic, confirming that the asymptotic theory of Section \ref{sec:WMDPDE} is valid at $K = 200$. Third, the asymptotic and bootstrap confidence intervals appear to agree closely throughout. This close agreement confirms the adequacy of the asymptotic approximation at the present sample size and supports the use of the simpler asymptotic intervals
	in practice when $K = 200$.
	
	Finally, CI widths increase only mildly with $\beta$, which is expected since the present dataset was generated under clean conditions. For the median, the transformed asymptotic width grows from $16{,}419$ cycles at $\beta = 0$ to $18{,}910$ cycles at $\beta = 1$, an increase of approximately $15\%$ over the full tuning range, a negligible penalty relative to the median lifetime of around $82{,}000$ cycles. Under clean data, estimators across the full range of $\beta$ perform similarly, with the MLE retaining a slight numerical advantage in efficiency. However, this modest loss in precision for larger $\beta$ is the cost of robustness assurance, as the simulation study of Section \ref{sec:sim_analysis} demonstrates; once contamination is introduced, the MLE deteriorates rapidly while estimators with larger $\beta$ remain stable. The small efficiency cost is therefore a worthwhile price to pay for robustness against data contamination.	
    
	\begin{figure}[htbp]
		\centering
		\includegraphics[width=0.9\textwidth]{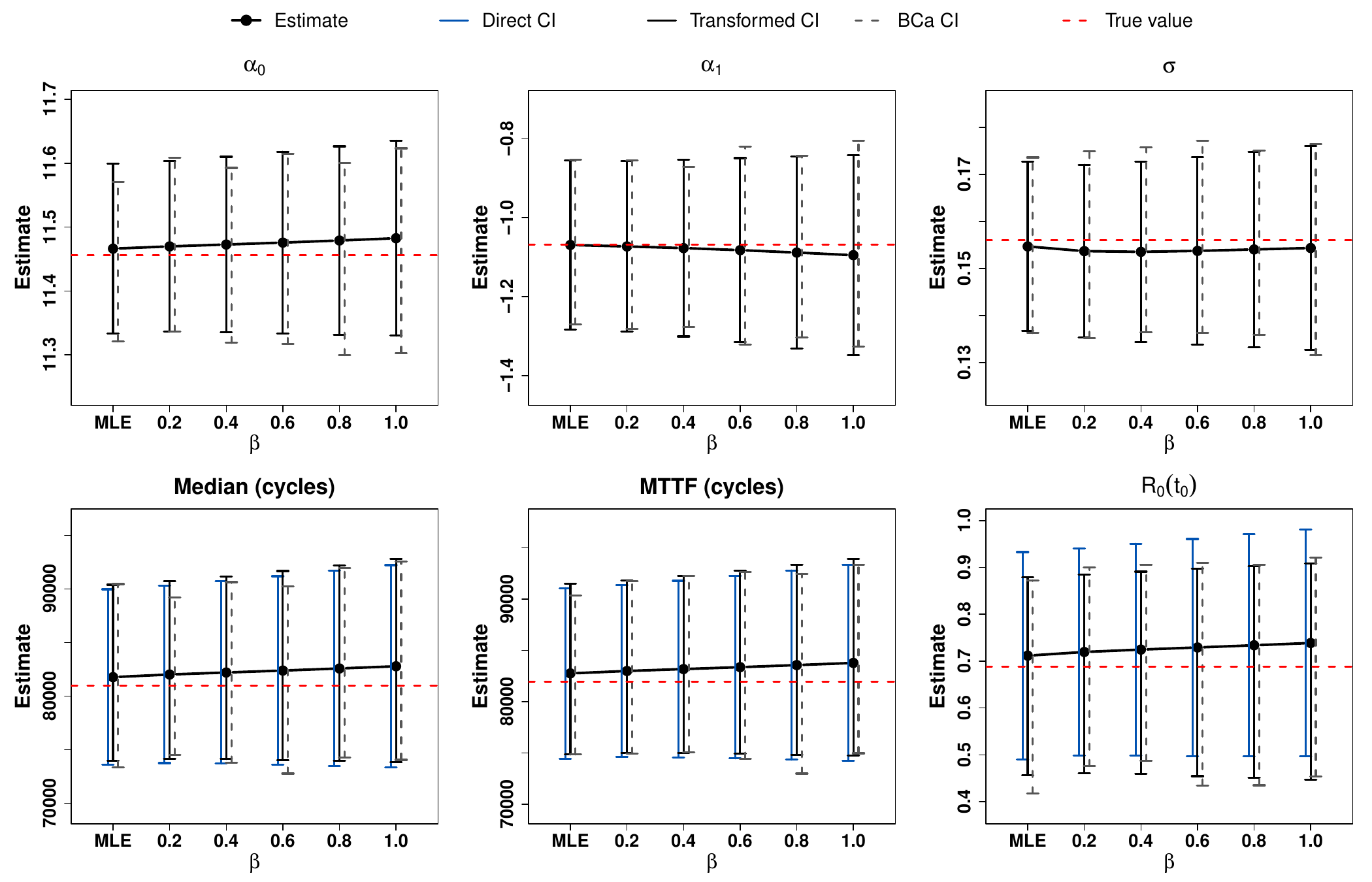}
		\caption{WMDPDE estimates, 95\% direct and transformed asymptotic confidence intervals, and 95\% BCa bootstrap confidence intervals for the model parameters $\alpha_0$, $\alpha_1$, $\sigma$ (top row) and lifetime characteristics at the use cyclic-stress condition (bottom row) for various $\beta$ values.}
		\label{fig:combined}
	\end{figure}
    
    Figure \ref{fig:combined} displays all six WMDPDE estimates with their $95\%$ asymptotic and BCa bootstrap confidence intervals while varying the tuning parameter $\beta$. The red dashed lines mark the true values. Three features are immediately apparent. First, all estimates are stable across the full range of $\beta$, with no systematic deterioration in point accuracy or interval width as robustness increases. Second, all asymptotic and bootstrap confidence intervals contain the true value in every panel, confirming the theoretical coverage guarantees derived in Section \ref{sec:WMDPDE}. Third, the asymptotic and bootstrap intervals are nearly indistinguishable, indicating that the covariance matrix $\widehat{\boldsymbol{\Sigma}}_\beta$ provides an accurate approximation to the sampling variability of the WMDPDE at $K = 200$. Therefore, the results here confirm the reliability and practical utility of the proposed methodology on a realistic engineering dataset.
	\FloatBarrier
	\section{Conclusion}
	\label{sec:conclusion}
    
    This paper has developed and illustrated robust inferential procedures for cyclic-stress accelerated life testing with lognormal lifetimes under interval monitoring and Type-I censoring. The WMDPDE, obtained by minimizing a weighted density power divergence over the $R$ independent stress groups, provides a flexible and efficient alternative to the MLE that is controlled by a single tuning parameter $\beta$. At $\beta = 0$ the WMDPDE reduces to the MLE, while increasing $\beta$ provides progressively stronger protection against data contamination at the cost of a small and well-controlled 
    efficiency loss. The asymptotic normality of the WMDPDE was developed under mild regularity conditions, enabling the construction of asymptotic confidence intervals for the assumed lognormal-CyALT model parameters and certain key lifetime characteristics, namely median lifetime, mean time to failure, and reliability at a specified mission time, using both the delta-method transformations and parametric BCa bootstrap procedures.
    
    The influence function analysis demonstrates that the WMDPDE exhibits improved robustness relative to the MLE for every $\beta > 0$, with full boundedness achieved for $\beta \geq 1$. The simulation study confirms that the WMDPDE behaves robustly under contamination. Even a small tuning value of $\beta = 0.2$ improves upon the MLE as soon as contamination is present, and larger values of $\beta$ offer stronger protection as contamination grows. The analysis of the automotive air-conditioner evaporator dataset confirms that all three types of confidence interval: direct asymptotic, transformed asymptotic, and BCa bootstrap, agree closely and contain the true parameter and WMDPD estimate values across the complete tuning parameter range, supporting the validity of the asymptotic theory at the sample sizes typical of industrial reliability experiments.
    
    There are several directions that remain open for future work. The present framework assumes lognormal lifetimes; extension to the Weibull distribution and log-location scale family, which is widely used in reliability engineering, is a natural next step. Divergence-based test hypothesis for the proposed WMDPDE in the CyALT setting, which are needed for formal comparison of reliability across stress conditions, are also of direct practical interest. In the end, extension to competing risks into consideration, where units may disfunction due to multiple causes, would also broaden the applicability of the proposed methods.
    
	\bibliographystyle{apalike}
    \bibliography{references}{}
	
	\appendix

\end{document}